\PassOptionsToPackage{dvipsnames,HTML}{xcolor}
\PassOptionsToPackage{sort&compress,capitalize,nameinlink}{cleveref}
\pdfoutput=1
\documentclass[a4paper,USenglish,cleveref]{lipics-v2021}
\hideLIPIcs

\title{Fully Fluctuating Sleepy Consensus from Minimal Assumptions}

\EventEditors{Ioannis Chatzigiannakis, Andrea Vitaletti, Keren Censor-Hillel, and William K. Moses Jr.}
\EventNoEds{4}
\EventLongTitle{40th International Symposium on Distributed Computing (DISC 2026)}
\EventShortTitle{DISC 2026}
\EventAcronym{DISC}
\EventYear{2026}
\EventDate{November 9--13, 2026}
\EventLocation{Rome, Italy}
\EventLogo{}
\SeriesVolume{397}
\ArticleNo{44}

\titlerunning{Fully Fluctuating Sleepy Consensus}

\author{Javier Nieto}{University of Illinois Urbana-Champaign, Urbana, IL, USA}{jmnieto2@illinois.edu}{https://orcid.org/0009-0003-8792-4010}{}
\author{Yuval Efron}{Institute for Advanced Study, Princeton, NJ, USA}{Efronyuv@ias.edu}{https://orcid.org/0000-0003-0882-9342}{}
\author{Joachim Neu}{a16z Crypto Research, New York, NY, USA}{jneu@a16z.com}{https://orcid.org/0000-0002-9777-6168}{}
\author{Ling Ren}{University of Illinois Urbana-Champaign, Urbana, IL, USA}{renling@illinois.edu}{https://orcid.org/0000-0003-3437-7570}{}

\authorrunning{J. Nieto, Y. Efron, J. Neu, and L. Ren}

\Copyright{Javier Nieto, Yuval Efron, Joachim Neu, and Ling Ren}



\funding{This work is funded in part by the National Science Foundation award \#2143058.}

\nolinenumbers %

\makeatletter
\let\lipics@orig@newtheorem\newtheorem
\def\newtheorem{\@ifnextchar*{\lipics@newtheorem@star}{\lipics@newtheorem@check}}
\def\lipics@newtheorem@star*#1{%
  \@ifundefined{#1}%
    {\lipics@orig@newtheorem*{#1}}%
    {\@gobble}%
}
\def\lipics@newtheorem@check#1{%
  \@ifundefined{#1}%
    {\lipics@orig@newtheorem{#1}}%
    {\lipics@newtheorem@gobble}%
}
\def\lipics@newtheorem@gobble#1{%
  \@ifnextchar[{\lipics@newtheorem@gobbleopt}{}%
}
\def\lipics@newtheorem@gobbleopt[#1]{}
\makeatother

\usepackage{subcaption}
\usepackage{array}
\usepackage{enumerate}
\usepackage{tabularx}
\usepackage{thmtools}
\usepackage{placeins}
\usepackage{./lib/shortsym}

\usepackage{import}
\usepackage[utf8]{inputenc}

\PassOptionsToPackage{hyphens}{url}%
\usepackage{graphicx}

\usepackage{amsmath}
\usepackage{amssymb}
\usepackage{amsfonts}
\usepackage{mathtools}

\usepackage{shortsym}

\usepackage{microtype}

\usepackage{IEEEtrantools}
\allowdisplaybreaks

\usepackage{import}
\theoremstyle{plain}
\newtheorem{lemma}{Lemma}
\newtheorem*{theorem*}{Theorem}
\newtheorem*{conjecture*}{Conjecture}

\theoremstyle{definition}
\makeatletter
\renewcommand*\env@matrix[1][*\c@MaxMatrixCols c]{%
    \hskip -\arraycolsep
    \let\@ifnextchar\new@ifnextchar
    \array{#1}}
\makeatother

\DeclarePairedDelimiter{\abs}{\lvert}{\rvert}
\DeclarePairedDelimiter{\len}{\lvert}{\rvert}
\DeclarePairedDelimiter{\norm}{\lVert}{\rVert}
\DeclarePairedDelimiter{\floor}{\lfloor}{\rfloor}
\DeclarePairedDelimiter{\ceil}{\lceil}{\rceil}

\makeatletter
\let\oldabs\abs
\def\abs{\@ifstar{\oldabs}{\oldabs*}}
\let\oldlen\len
\def\len{\@ifstar{\oldlen}{\oldlen*}}
\let\oldnorm\norm
\def\norm{\@ifstar{\oldnorm}{\oldnorm*}}
\let\oldfloor\floor
\def\floor{\@ifstar{\oldfloor}{\oldfloor*}}
\let\oldceil\ceil
\def\ceil{\@ifstar{\oldceil}{\oldceil*}}
\makeatother

\usepackage{pifont}

\usepackage{dsfont}

\usepackage{xspace}

\usepackage{xcolor} %

\definecolor{myA16zGrayLight}{RGB}{235,235,235}     %
\definecolor{myA16zGrayMedium}{RGB}{196,196,196}    %
\definecolor{myA16zGrayDark}{RGB}{44,34,34}         %
\definecolor{myA16zLavender}{RGB}{208,161,255}      %
\definecolor{myA16zMagenta}{RGB}{195,70,206}        %
\definecolor{myA16zMulberry}{RGB}{113,24,88}        %
\definecolor{myA16zLemonChiffon}{RGB}{250,234,157}  %
\definecolor{myA16zAmber}{RGB}{230,154,48}          %
\definecolor{myA16zRust}{RGB}{174,59,10}            %
\definecolor{myA16zLime}{RGB}{197,222,107}          %
\definecolor{myA16zAquamarine}{RGB}{82,216,145}     %
\definecolor{myA16zPine}{RGB}{60,87,44}             %
\definecolor{myA16zPacific}{RGB}{145,224,235}       %
\definecolor{myA16zTeal}{RGB}{36,197,201}           %
\definecolor{myA16zAzure}{RGB}{18,51,90}            %

\definecolor{myTechnionDeepBlue}{HTML}{002147}           %
\definecolor{myTechnionGoldenOchre}{HTML}{D59F0F}        %
\definecolor{myTechnionBlack}{HTML}{000000}              %
\definecolor{myTechnionWhite}{HTML}{FFFFFF}              %

\definecolor{myTechnionRed}{HTML}{E31D1A}             %
\definecolor{myTechnionPink}{HTML}{EA094B}           %
\definecolor{myTechnionPurple1}{HTML}{AE3B72}         %
\definecolor{myTechnionPurple2}{HTML}{4D4084}         %
\definecolor{myTechnionBlue1}{HTML}{216093}            %
\definecolor{myTechnionBlue2}{HTML}{5686DA}           %
\definecolor{myTechnionTeal}{HTML}{32B1CA}            %
\definecolor{myTechnionGreen1}{HTML}{EA094B}           %
\definecolor{myTechnionGreen2}{HTML}{A3D65C}           %
\definecolor{myTechnionGreen3}{HTML}{94D60A}           %
\definecolor{myTechnionYellow}{HTML}{FDD700}          %
\definecolor{myTechnionOrange}{HTML}{FF6B00}         %
\definecolor{myTechnionBrown}{HTML}{97775C}          %
\definecolor{myTechnionBeige}{HTML}{D9D1C3}          %
\definecolor{myTechnionGray1}{HTML}{A2A9AE}            %
\definecolor{myTechnionGray2}{HTML}{5A6771}            %

\definecolor{mySuCardinalRed}{HTML}{8c1515}
\definecolor{mySuCardinalRedLight}{HTML}{B83A4B}
\definecolor{mySuCardinalRedDark}{HTML}{820000}
\definecolor{mySuWhite}{HTML}{ffffff}
\definecolor{mySuCoolGrey}{HTML}{53565A}
\definecolor{mySuBlack}{HTML}{2e2d29}
\definecolor{mySuBlack100}{HTML}{2e2d29}
\definecolor{mySuBlack90}{HTML}{43423E}
\definecolor{mySuBlack80}{HTML}{585754}
\definecolor{mySuBlack70}{HTML}{6D6C69}
\definecolor{mySuBlack60}{HTML}{767674}
\definecolor{mySuBlack50}{HTML}{979694}
\definecolor{mySuBlack40}{HTML}{ABABA9}
\definecolor{mySuBlack30}{HTML}{C0C0BF}
\definecolor{mySuBlack20}{HTML}{D5D5D4}
\definecolor{mySuBlack10}{HTML}{EAEAEA}

\definecolor{mySuPaloAlto}{HTML}{175E54}
\definecolor{mySuPaloAltoLight}{HTML}{2D716F}
\definecolor{mySuPaloAltoDark}{HTML}{014240}
\definecolor{mySuPaloVerde}{HTML}{279989}
\definecolor{mySuPaloVerdeLight}{HTML}{59B3A9}
\definecolor{mySuPaloVerdeDark}{HTML}{017E7C}
\definecolor{mySuOlive}{HTML}{8F993E}
\definecolor{mySuOliveLight}{HTML}{A6B168}
\definecolor{mySuOliveDark}{HTML}{7A863B}
\definecolor{mySuBay}{HTML}{6FA287}
\definecolor{mySuBayLight}{HTML}{8AB8A7}
\definecolor{mySuBayDark}{HTML}{417865}
\definecolor{mySuSky}{HTML}{4298B5}
\definecolor{mySuSkyLight}{HTML}{67AFD2}
\definecolor{mySuSkyDark}{HTML}{016895}
\definecolor{mySuLagunita}{HTML}{007C92}
\definecolor{mySuLagunitaLight}{HTML}{009AB4}
\definecolor{mySuLagunitaDark}{HTML}{006B81}
\definecolor{mySuPoppy}{HTML}{E98300}
\definecolor{mySuPoppyLight}{HTML}{F9A44A}
\definecolor{mySuPoppyDark}{HTML}{D1660F}
\definecolor{mySuSpirited}{HTML}{E04F39}
\definecolor{mySuSpiritedLight}{HTML}{F4795B}
\definecolor{mySuSpiritedDark}{HTML}{C74632}
\definecolor{mySuIlluminating}{HTML}{FEDD5C}
\definecolor{mySuIlluminatingLight}{HTML}{FFE781}
\definecolor{mySuIlluminatingDark}{HTML}{FEC51D}
\definecolor{mySuPlum}{HTML}{620059}
\definecolor{mySuPlumLight}{HTML}{734675}
\definecolor{mySuPlumDark}{HTML}{350D36}
\definecolor{mySuBrick}{HTML}{651C32}
\definecolor{mySuBrickLight}{HTML}{7F2D48}
\definecolor{mySuBrickDark}{HTML}{42081B}
\definecolor{mySuArchway}{HTML}{5D4B3C}
\definecolor{mySuArchwayLight}{HTML}{766253}
\definecolor{mySuArchwayDark}{HTML}{2F2424}
\definecolor{mySuStone}{HTML}{7F7776}
\definecolor{mySuStoneLight}{HTML}{D4D1D1}
\definecolor{mySuStoneDark}{HTML}{544948}
\definecolor{mySuFog}{HTML}{DAD7CB}
\definecolor{mySuFogLight}{HTML}{F4F4F4}
\definecolor{mySuFogDark}{HTML}{B6B1A9}

\definecolor{mySuDigitalRed}{HTML}{B1040E}
\definecolor{mySuDigitalRedLight}{HTML}{E50808}
\definecolor{mySuDigitalRedDark}{HTML}{820000}
\definecolor{mySuDigitalBlue}{HTML}{006CB8}
\definecolor{mySuDigitalBlueLight}{HTML}{6FC3FF}
\definecolor{mySuDigitalBlueDark}{HTML}{00548f}
\definecolor{mySuDigitalGreen}{HTML}{008566}
\definecolor{mySuDigitalGreenLight}{HTML}{1AECBA}
\definecolor{mySuDigitalGreenDark}{HTML}{006F54}

\definecolor{myParula1Blue}{RGB}{0,114,189}
\definecolor{myParula2Orange}{RGB}{217,83,25}
\definecolor{myParula3Yellow}{RGB}{237,177,32}
\definecolor{myParula4Purple}{RGB}{126,47,142}
\definecolor{myParula5Green}{RGB}{119,172,48}
\definecolor{myParula6LightBlue}{RGB}{77,190,238}
\definecolor{myParula7Red}{RGB}{162,20,47}

\usepackage{multirow}
\usepackage{booktabs}   %
\usepackage{makecell}
\usepackage{threeparttable}   %

\usepackage{algorithm}
\usepackage{algorithmicx}
\usepackage[noend]{algpseudocode}

\makeatletter
\AddToHook{env/algorithmic/begin}{\def\@currentcounter{ALG@line}}
\renewcommand{\theHALG@line}{\thealgorithm.\arabic{ALG@line}}
\makeatother

\algnewcommand{\LineComment}[1]{\State {\textcolor{gray}{/\!/ #1}}}

\algrenewcommand{\alglinenumber}[1]{\scriptsize\textcolor{gray}{\texttt{#1}}}
\algrenewcommand{\algorithmicindent}{1em}

\algnewcommand{\algfontsize}[0]{}
\AddToHook{env/algorithmic/begin}{\algfontsize}

\algnewcommand{\algorithmicswitch}{\textbf{switch}}
\algdef{SE}[SWITCH]{Switch}{EndSwitch}[1]{\algorithmicswitch\ #1\ \algorithmicdo}{\algorithmicend\ \algorithmicswitch}%
\algtext*{EndSwitch}%

\algnewcommand{\algorithmiccase}{\textbf{case}}
\algdef{SE}[CASE]{Case}{EndCase}[1]{\algorithmiccase\ #1}{\algorithmicend\ \algorithmiccase}%
\algtext*{EndCase}%

\algnewcommand{\algorithmicon}{\textbf{on}}
\algdef{SE}[ON]{On}{EndOn}[1]{\algorithmicon\ #1\ \algorithmicdo}{\algorithmicend\ \algorithmicon}%
\algtext*{EndOn}%

\algnewcommand{\algorithmicat}{\textbf{at}}
\algdef{SE}[AT]{At}{EndAt}[1]{\algorithmicat\ #1\ \algorithmicdo}{\algorithmicend\ \algorithmicat}%
\algtext*{EndAt}%

\algnewcommand{\algorithmicrealfunction}{\textbf{function}}
\algdef{SE}[REALFUNCTION]{RealFunction}{EndRealFunction}[1]{\algorithmicrealfunction\ #1\ \algorithmicdo}{\algorithmicend\ \algorithmicrealfunction}%
\algtext*{EndRealFunction}%

\algnewcommand{\algorithmicthroughout}{\textbf{do throughout}}
\algdef{SE}[Throughout]{Throughout}{EndThroughout}[1]{\algorithmicthroughout\ #1\ \algorithmicdo}{\algorithmicend\ \algorithmicthroughout}%
\algtext*{EndThroughout}%

\algnewcommand{\algorithmictry}{\textbf{try}}
\algnewcommand{\algorithmiccatch}{\textbf{catch}}
\algdef{SE}[TRY]{Try}{EndTry}{\algorithmictry\ \algorithmicdo}{\algorithmicend\ \algorithmictry}
\algdef{C}[TRY]{TRY}{Catch}[1]{\algorithmiccatch\ #1\ \algorithmicdo}
\algtext*{EndTry}

\algrenewcommand{\algorithmicdo}{}
\algrenewcommand{\algorithmicthen}{}

\algnewcommand{\algorithmicgoto}{\textbf{goto}}%
\algnewcommand{\Goto}[1]{\algorithmicgoto~\ref{#1}}%

\algnewcommand{\algorithmicassert}{\textbf{assert}}%
\algnewcommand{\Assert}[1]{\algorithmicassert~{#1}}%

\algnewcommand{\algorithmicbreak}{\textbf{break}}%
\algnewcommand{\Break}[0]{\algorithmicbreak}%
\algnewcommand{\BreakOutOf}[1]{\algorithmicbreak~out~of~#1}%

\algnewcommand{\algorithmicwaiton}{\textbf{wait on}}%
\algnewcommand{\WaitOn}[1]{\algorithmicwaiton~{#1}}%

\algnewcommand{\InlineRequire}[1]{\textbf{require} {#1}}

\algblock{ManualIndent}{EndManualIndent}
\algnotext{ManualIndent}
\algnotext{EndManualIndent}

\algdef{SE}[GENERICBLOCK]{GenericBlock}{EndGenericBlock}[1]{#1}{}%
\algtext*{EndGenericBlock}%

\AtBeginEnvironment{appendices}{%
    \crefalias{section}{appendix}%
    \crefalias{subsection}{subappendix}%
    \crefalias{subsubsection}{subsubappendix}%
    \crefalias{subsubsubsection}{subsubsubappendix}%
}
\AddToHook{cmd/appendix/after}{%
    \crefalias{section}{appendix}%
    \crefalias{subsection}{subappendix}%
    \crefalias{subsubsection}{subsubappendix}%
    \crefalias{subsubsubsection}{subsubsubappendix}%
}

\hypersetup{hypertexnames=false}

\crefalias{ALG@line}{line}

\crefname{figure}{Fig.}{Figs.}
\Crefname{figure}{Fig.}{Figs.}

\crefname{table}{Tab.}{Tabs.}
\Crefname{table}{Tab.}{Tabs.}

\crefname{section}{Sec.}{Secs.}
\Crefname{section}{Sec.}{Secs.}
\crefname{subsection}{Sec.}{Secs.}
\Crefname{subsection}{Sec.}{Secs.}
\crefname{subsubsection}{Sec.}{Secs.}
\Crefname{subsubsection}{Sec.}{Secs.}
\crefname{subsubsubsection}{Sec.}{Secs.}
\Crefname{subsubsubsection}{Sec.}{Secs.}
\crefname{appendix}{App.}{Apps.}
\Crefname{appendix}{App.}{Apps.}
\crefname{subappendix}{App.}{Apps.}
\Crefname{subappendix}{App.}{Apps.}
\crefname{subsubappendix}{App.}{Apps.}
\Crefname{subsubappendix}{App.}{Apps.}
\crefname{subsubsubappendix}{App.}{Apps.}
\Crefname{subsubsubappendix}{App.}{Apps.}

\crefname{algorithm}{Alg.}{Algs.}
\Crefname{algorithm}{Alg.}{Algs.}
\crefname{line}{ln.}{lns.}
\Crefname{line}{ln.}{lns.}

\crefname{proposition}{Prop.}{Props.}
\Crefname{proposition}{Prop.}{Props.}
\crefname{lemma}{Lem.}{Lems.}
\Crefname{lemma}{Lem.}{Lems.}
\crefname{theorem}{Thm.}{Thms.}
\Crefname{theorem}{Thm.}{Thms.}
\crefname{corollary}{Cor.}{Cors.}
\Crefname{corollary}{Cor.}{Cors.}
\crefname{definition}{Def.}{Defs.}
\Crefname{definition}{Def.}{Defs.}
\crefname{observation}{Obs.}{Obs.}
\Crefname{observation}{Obs.}{Obs.}
\usepackage{tikz}
\usetikzlibrary{fit}
\usetikzlibrary{math}
\usetikzlibrary{calc}
\usetikzlibrary{positioning}
\usetikzlibrary{decorations.pathmorphing}
\usetikzlibrary{decorations.pathreplacing}
\usetikzlibrary{backgrounds}
\usetikzlibrary{patterns}
\usetikzlibrary{matrix}
\usetikzlibrary{arrows.meta}
\usetikzlibrary{shapes.callouts}

\makeatletter
\NewDocumentCommand {\getnodedimen} {O{\nodewidth} O{\nodeheight} m} {
  \begin{pgfinterruptboundingbox}
    \begin{scope}[local bounding box=bb@temp]
      \node[inner sep=0pt, fit=(#3)] {};
    \end{scope}
    \path ($(bb@temp.north east)-(bb@temp.south west)$);
  \end{pgfinterruptboundingbox}
  \pgfgetlastxy{#1}{#2}
}
\makeatother

\AtBeginDocument{\sbox0{$x$}}%
\tikzset{baseshift/.style={yshift=-\the\dimexpr\fontdimen22\textfont2}}
\tikzset{
  MAT/.style={ampersand replacement=\&},
  XS/.style={scale=0.6},
  SM/.style={scale=0.75},
  MD/.style={scale=1.25},
  LG/.style={scale=1.5},
  MEMBER/.style={draw, circle, very thick, minimum size=4mm, font=\normalsize},
  ADV/.style={draw=red!60, text=red},
  ASLEEP/.style={opacity=0.55, densely dotted},
  BOOTING/.style={draw=black!70, dashed},
  NONMEMBER/.style={draw=red!70!gray, text=red!70!gray, dotted, minimum size=4mm, font=\normalsize},
  VOTESET/.style={draw, rounded corners=3pt, inner xsep=1.5mm, inner ysep=0.75mm, thick},
  EQUIVSET/.style={VOTESET, dashed, draw=red!70, thin},
  GOODBOX/.style={draw=none, fill=green!80!black!50, opacity=0.2, rounded corners=5pt, inner xsep=3mm, inner ysep=1mm},
  BADBOX/.style={draw=none, fill=red!75, opacity=0.2, rounded corners=5pt, inner xsep=3mm, inner ysep=1mm},
  RANGE/.style={-latex, draw=red!60, opacity=0.75, line width=0.25mm},
  SIM/.style={-{Computer Modern Rightarrow[width=1.5mm]}, draw=red!60, opacity=0.75, line width=0.25mm, decorate, decoration={snake, amplitude=0.25mm, segment length=3mm}},
  MSGSENT/.style={rectangle callout, rounded corners=2pt, callout pointer width=1.5mm, inner xsep=1mm, inner ysep=1mm, font=\small, fill=white, draw=gray!60},
  MSGRECV/.style={cloud callout, cloud puffs=8, cloud puff arc=110, callout pointer segments=2, inner xsep=-0.75mm, inner ysep=-0.5mm, font=\small, fill=white, draw=gray!60},
  LABELCALLOUT/.style={rectangle callout, rounded corners=1pt, callout pointer width=0.75mm, inner xsep=1.25mm, inner ysep=0.75mm, font=\normalsize, fill=white, draw=black, thick},
  CALLOUT/.style={draw, thick, rounded corners, font=\scriptsize, align=center},
}
\subimport{./lib/}{defer.tex}
\AddToHook{env/lemma/begin}{\crefalias{theorem}{lemma}}
\AddToHook{env/corollary/begin}{\crefalias{theorem}{corollary}}
\AddToHook{env/proposition/begin}{\crefalias{theorem}{proposition}}
\AddToHook{env/exercise/begin}{\crefalias{theorem}{exercise}}
\AddToHook{env/definition/begin}{\crefalias{theorem}{definition}}
\AddToHook{env/conjecture/begin}{\crefalias{theorem}{conjecture}}
\AddToHook{env/observation/begin}{\crefalias{theorem}{observation}}
\AddToHook{env/example/begin}{\crefalias{theorem}{example}}
\AddToHook{env/note/begin}{\crefalias{theorem}{note}}
\AddToHook{env/remark/begin}{\crefalias{theorem}{remark}}
\AddToHook{env/claim/begin}{\crefalias{theorem}{claim}}
\usepackage{tcolorbox} %
\crefalias{ALG@line}{line}

\crefname{figure}{Fig.}{Figs.}
\Crefname{figure}{Fig.}{Figs.}

\crefname{table}{Tab.}{Tabs.}
\Crefname{table}{Tab.}{Tabs.}

\crefname{section}{Sec.}{Secs.}
\Crefname{section}{Sec.}{Secs.}
\crefname{subsection}{Sec.}{Secs.}
\Crefname{subsection}{Sec.}{Secs.}
\crefname{subsubsection}{Sec.}{Secs.}
\Crefname{subsubsection}{Sec.}{Secs.}
\crefname{subsubsubsection}{Sec.}{Secs.}
\Crefname{subsubsubsection}{Sec.}{Secs.}
\crefname{appendix}{App.}{Apps.}
\Crefname{appendix}{App.}{Apps.}
\crefname{subappendix}{App.}{Apps.}
\Crefname{subappendix}{App.}{Apps.}
\crefname{subsubappendix}{App.}{Apps.}
\Crefname{subsubappendix}{App.}{Apps.}
\crefname{subsubsubappendix}{App.}{Apps.}
\Crefname{subsubsubappendix}{App.}{Apps.}

\crefname{algorithm}{Alg.}{Algs.}
\Crefname{algorithm}{Alg.}{Algs.}
\crefname{line}{ln.}{lns.}
\Crefname{line}{ln.}{lns.}

\crefname{proposition}{Prop.}{Props.}
\Crefname{proposition}{Prop.}{Props.}
\crefname{lemma}{Lem.}{Lems.}
\Crefname{lemma}{Lem.}{Lems.}
\crefname{theorem}{Thm.}{Thms.}
\Crefname{theorem}{Thm.}{Thms.}
\crefname{corollary}{Cor.}{Cors.}
\Crefname{corollary}{Cor.}{Cors.}
\crefname{definition}{Def.}{Defs.}
\Crefname{definition}{Def.}{Defs.}
\crefname{conjecture}{Conj.}{Conjs.}
\Crefname{conjecture}{Conj.}{Conjs.}
\crefname{remark}{Rem.}{Rems.}
\Crefname{remark}{Rem.}{Rems.}

\usepackage{subcaption}
\usepackage{fontawesome5}
\usepackage{pifont}
\usepackage{nicefrac}

\usepackage{xspace}

\MakeRobust{\Call}

\newcommand{\code}[1]{\Call{#1}{}}
\newcommand{\authcite}[1]{%
  \ifcsname authname@#1\endcsname
    \csname authname@#1\endcsname~\cite{#1}%
  \else
    \GenericError{}{authcite: no generated name for key '#1' (run make authnames)}{}{}%
  \fi
}

\expandafter\def\csname authname@ABY22\endcsname{Abraham et al.}%
\expandafter\def\csname authname@ADKS22\endcsname{Abraham et al.}%
\expandafter\def\csname authname@AKSW22\endcsname{Attiya et al.}%
\expandafter\def\csname authname@AMN+20\endcsname{Abraham et al.}%
\expandafter\def\csname authname@AND97\endcsname{Anderson}%
\expandafter\def\csname authname@BBBF18\endcsname{Boneh et al.}%
\expandafter\def\csname authname@BDNP14\endcsname{Bonnet et al.}%
\expandafter\def\csname authname@BFFPT25\endcsname{Bonomi et al.}%
\expandafter\def\csname authname@BFT23\endcsname{Bonomi et al.}%
\expandafter\def\csname authname@BGH95\endcsname{Buhrman et al.}%
\expandafter\def\csname authname@BGK+18\endcsname{Badertscher et al.}%
\expandafter\def\csname authname@BGP89\endcsname{Berman et al.}%
\expandafter\def\csname authname@BHK+20\endcsname{Buterin et al.}%
\expandafter\def\csname authname@BM99\endcsname{Bellare and Miner}%
\expandafter\def\csname authname@BSA14\endcsname{Bessani et al.}%
\expandafter\def\csname authname@BSIW12\endcsname{Banu et al.}%
\expandafter\def\csname authname@CASV95\endcsname{Cristian et al.}%
\expandafter\def\csname authname@CGG+21\endcsname{Choudhuri et al.}%
\expandafter\def\csname authname@CL99\endcsname{Castro and Liskov}%
\expandafter\def\csname authname@CM19\endcsname{Chen and Micali}%
\expandafter\def\csname authname@CMSK07\endcsname{Chun et al.}%
\expandafter\def\csname authname@DDG+23\endcsname{David et al.}%
\expandafter\def\csname authname@DGKR17\endcsname{David et al.}%
\expandafter\def\csname authname@DGNW19\endcsname{Drijvers et al.}%
\expandafter\def\csname authname@DKT21\endcsname{Deb et al.}%
\expandafter\def\csname authname@DNTT23\endcsname{D'Amato et al.}%
\expandafter\def\csname authname@DPP19\endcsname{Deirmentzoglou et al.}%
\expandafter\def\csname authname@DPS16\endcsname{Daian et al.}%
\expandafter\def\csname authname@DSTZ24\endcsname{D'Amato et al.}%
\expandafter\def\csname authname@DZ22\endcsname{Duan and Zhang}%
\expandafter\def\csname authname@Dol82\endcsname{Dolev}%
\expandafter\def\csname authname@ENP25\endcsname{Efron et al.}%
\expandafter\def\csname authname@ENRT26\endcsname{Efron et al.}%
\expandafter\def\csname authname@ET25\endcsname{Efron and Tas}%
\expandafter\def\csname authname@FLPA26\endcsname{Farahbakhsh et al.}%
\expandafter\def\csname authname@G94\endcsname{Garay}%
\expandafter\def\csname authname@GDCL22\endcsname{Gao et al.}%
\expandafter\def\csname authname@GHK+21\endcsname{Gentry et al.}%
\expandafter\def\csname authname@GHM+17a\endcsname{Gilad et al.}%
\expandafter\def\csname authname@GKL15\endcsname{Garay et al.}%
\expandafter\def\csname authname@GL02\endcsname{Gilbert and Lynch}%
\expandafter\def\csname authname@GL23\endcsname{Gafni and Losa}%
\expandafter\def\csname authname@GP92\endcsname{Garay and Perry}%
\expandafter\def\csname authname@HM90\endcsname{Halpern and Moses}%
\expandafter\def\csname authname@IR01\endcsname{Itkis and Reyzin}%
\expandafter\def\csname authname@JM14\endcsname{Jehl and Meling}%
\expandafter\def\csname authname@KRDO17\endcsname{Kiayias et al.}%
\expandafter\def\csname authname@LAB+06\endcsname{Lorch et al.}%
\expandafter\def\csname authname@LMZ09\endcsname{Lamport et al.}%
\expandafter\def\csname authname@LMZ10\endcsname{Lamport et al.}%
\expandafter\def\csname authname@LR23\endcsname{Lewis-Pye and Roughgarden}%
\expandafter\def\csname authname@LR23b\endcsname{Lewis-Pye and Roughgarden}%
\expandafter\def\csname authname@LR24\endcsname{Budish et al.}%
\expandafter\def\csname authname@LSP82\endcsname{Lamport et al.}%
\expandafter\def\csname authname@LSS24\endcsname{Loss et al.}%
\expandafter\def\csname authname@MMR23\endcsname{Malkhi et al.}%
\expandafter\def\csname authname@MNR19\endcsname{Malkhi et al.}%
\expandafter\def\csname authname@MR22\endcsname{Momose and Ren}%
\expandafter\def\csname authname@MRV99\endcsname{Micali et al.}%
\expandafter\def\csname authname@NNR26\endcsname{Nieto et al.}%
\expandafter\def\csname authname@NTT20a\endcsname{Neu et al.}%
\expandafter\def\csname authname@Nak09\endcsname{Nakamoto}%
\expandafter\def\csname authname@OO14\endcsname{Ongaro and Ousterhout}%
\expandafter\def\csname authname@OY91\endcsname{Ostrovsky and Yung}%
\expandafter\def\csname authname@PP25\endcsname{Pan and Potop-Butucaru}%
\expandafter\def\csname authname@PS17\endcsname{Pass and Shi}%
\expandafter\def\csname authname@R85\endcsname{Reischuk}%
\expandafter\def\csname authname@SRMJ12\endcsname{Shraer et al.}%
\expandafter\def\csname authname@SYKM13\endcsname{Sasaki et al.}%
\expandafter\def\csname authname@TSZ23\endcsname{Tzinas et al.}%
\expandafter\def\csname authname@TY26\endcsname{Tang and Ye}%
\expandafter\def\csname authname@WTW+24\endcsname{Wei et al.}%

\newcolumntype{Y}{>{\centering\arraybackslash}X}

\newcommand{\ba}{\mathsf{BA}}
\newcommand{\ga}{\mathsf{GA}}

\newcommand{\gw}{\mathsf{GW}}

\newenvironment{protocolsteps}{%
  \begin{list}{}{%
    \setlength{\leftmargin}{1.75em}%
    \setlength{\labelsep}{0.5em}%
    \setlength{\labelwidth}{\dimexpr\leftmargin-\labelsep\relax}%
  }%
}{\end{list}}

\newcommand{\nodes}{\mathcal{P}}

\newcommand{\hon}{\mathcal{H}}
\newcommand{\hons}[2]{\hon_{\awake}({#1}, {#2})}

\newcommand{\adv}{\mathcal{A}}
\newcommand{\advs}[2]{\adv_{\awake}({#1}, {#2})}

\newcommand{\awake}{\mathcal{W}}

\newcommand{\ins}{\mathcal{I}}
\newcommand{\valspace}{\mathcal{V}}

\newcommand{\sbl}{T_{\mathrm{s}}}

\newcommand{\rec}{T_{\mathrm{r}}}
\newcommand{\gmax}{g_{\mathrm{max}}}
\newcommand{\stable}[1][$\start$]{{#1}-stable\xspace}
\newcommand{\vstable}[1][$\start$]{{#1}-very-stable\xspace}

\newcommand{\negl}[1]{\Call{negl}{#1}}

\newcommand{\oracle}{\mathcal{O}}

\newcommand{\osig}{\oracle_{\code{S}}}

\newcommand{\ovrf}{\oracle_{\code{V}}}

\newcommand{\vrf}{\mathsf{VRF}}
\newcommand{\vrfo}[2][\sk]{\vrf_{#1}(#2)}

\newcommand{\sk}{\code{sk}}

\newcommand{\signed}[2][]{\left\langle #2 \right\rangle_{#1}}

\newcommand{\nonce}{r}
\newcommand{\nonces}{R}

\newcommand{\equivs}{E}
\newcommand{\gawake}[1][]{W_{#1}}
\newcommand{\noinput}{\emptyset}
\newcommand{\inputters}{I}

\newcommand{\editcolor}{red}
\newcommand{\edit}[1]{\textcolor{\editcolor}{#1}}
\newcommand{\medit}[1]{\mathcolor{\editcolor}{#1}}

\newcommand{\proposemsgtype}{\mathsf{PROPOSE}}

\newcommand{\signedproposemsg}[2][]{\signed[#1]{\proposemsgtype, #2}}

\newcommand{\inputmsgtype}{\mathsf{IN}}

\newcommand{\signedinputmsg}[2][]{\signed[#1]{\inputmsgtype, #2}}

\newcommand{\pingmsgtype}{\mathsf{PING}}

\newcommand{\signedpingmsg}[2][]{\signed[#1]{\pingmsgtype, #2}}
\newcommand{\ackmsgtype}{\mathsf{ACK}}

\newcommand{\signedackmsg}[2][]{\signed[#1]{\ackmsgtype, #2}}

\newcommand{\poly}{\ensuremath{\operatorname{poly}}}

\DeclareRobustCommand{\honpic}{\tikz[baseline=(current bounding box.base)]{\node[MEMBER, XS, yshift=3] {};}\xspace}

\DeclareRobustCommand{\advpic}{\tikz[baseline=(current bounding box.base)]{\node[MEMBER, XS, ADV, yshift=3] {};}\xspace}

\newdefergroup{ga2proofs}[notdeferred]
\newdefergroup{ga3proofs}[notdeferred]
\newdefergroup{baproofs}[notdeferred]

\begin{document}%
\maketitle%
\begin{abstract}
    Bitcoin's proof-of-work (PoW)-based protocol is remarkable for how little it asks of its participants.
    Not only can miners take breaks from work whenever they please, but it is almost unique in offering a path of contrition: corrupt miners can reclaim honest status simply by resuming mining on the longest chain.
    The protocol only requires that honest miners hold the majority of computational power at any given time.
    Analogous proof-of-stake (PoS) protocols, usually formalized via the sleepy model of~\authcite{PS17}, have fallen short of matching this robustness.
    In fact, sleepy consensus protocols in the plain PKI model \emph{must} heavily restrict fluctuations in adversarial participation over time.
    The recent work of~\authcite{ENP25} enables fully fluctuating participation in the sleepy model by introducing the external adversary model.
    Their protocol, however, relies on verifiable delay functions (VDFs), a strong cryptographic primitive that
    somewhat resembles PoW, by assuming that the adversary cannot compute \emph{sequential} work significantly faster than honest nodes.

    In this work, we design a sleepy consensus protocol for fully fluctuating participation with an external adversary under an honest majority, from minimal assumptions: a public key infrastructure (PKI) and a verifiable random function (VRF).
    In particular, we make no VDF or hardware assumptions.
    Our key technique is \emph{graded wakeness}, a novel primitive that allows nodes to form consistent opinions on which other nodes are awake.
    We further extend our protocol to handle \emph{uncorruption}, where corrupt nodes return to honesty.
    This extension requires only a mild additional assumption on the unpredictability of VRF outputs for liveness.

\end{abstract}%

\section{Introduction}
\label{intro}

The \emph{Byzantine agreement} problem was first described by~\authcite{LSP82} as a group of generals who must reach consensus on a battle plan, despite the presence of traitors among them.
This abstraction laid the foundation for fault-tolerant distributed algorithms for decades.
The classic setting appeals to our intuition via an analogy to a standing army, in which participation is mandatory.

Modern decentralized systems, such as permissionless blockchains, defy this rigid structure.
Bitcoin's proof-of-work (PoW)~\cite{Nak09}, for instance, is indifferent to how many miners are participating.
Here, miners more closely resemble a group of volunteers, where the system proceeds with whoever happens to be participating at any given time.
Accordingly, Bitcoin has accommodated growth from just a handful of miners to many thousands, as well as occasional notable declines due to various factors such as outages, government regulations, block reward ``halving,'' or price crashes.
The same flexibility naturally applies to the adversary: sustaining adversarial control is costly, so corrupt miners, too, may take a break.
In summary, honest and adversarial nodes alike may fluctuate in their participation.
We call this model \emph{fully fluctuating participation}.\footnote{Like most prior works in the literature, we assume a static membership set of eligible participants throughout and treat \emph{reconfiguration} (i.e., membership changes over time) as an orthogonal problem.}

The recent work of~\authcite{ENP25} introduces a model of fully fluctuating participation and gives the first protocol provably secure in this model without resorting to PoW.
Their protocol, however, replaces PoW with verifiable delay functions (VDFs)~\cite{BBBF18}, which are proofs of sequential work assumed to require approximately the same time to compute for any node, honest or adversarial.
Observe that a VDF is not only a cryptographic assumption, but also a strong and debatable hardware assumption: that the adversary cannot acquire or build any alternative hardware platform to speed up the VDF computation.

This leaves a fundamental question open: \emph{can consensus under fully fluctuating participation be achieved without any hardware assumptions such as PoW or VDFs?}

In this work, we answer this open question affirmatively and present a consensus protocol secure under fully fluctuating participation. %
Our protocol assumes only a standard public key infrastructure (PKI) setup and a well-established cryptographic primitive called a verifiable random function (VRF)~\cite{MRV99}.

We also extend the fully fluctuating participation model to allow \emph{uncorruption}, i.e., a corrupt node can return to honesty and restore the corruption budget.
As with participation, Bitcoin already exhibits this.
Once the adversary can no longer sustain control, the miner may simply resume work on the canonical chain as honest miners do.
Our protocol cleanly extends to this more general setting with only minor changes and a mild additional assumption on the unpredictability of VRF outputs for liveness.

\subsection{Technical Overview}
\label{intro:technical}
In the conventional adversarial model, a central adversary seizes the entire state of corrupt nodes, including their cryptographic secrets (private signing keys).
This model immediately rules out the possibility of corrupt nodes taking a break or returning to honesty, since the adversary can indefinitely impersonate them using their signing keys.
Prior works under this model must then restrict honest nodes not on break to always outnumber all nodes \emph{ever} corrupted since the start of execution.

To enable corrupt nodes to take breaks,~\authcite{ENP25} introduce the \emph{external adversary model}, which treats the node's cryptographic secrets as external to the protocol, just like mining hardware in Bitcoin.
This model is well justified because, in practice, corrupt nodes may perform actions on the adversary's behalf but are unlikely to hand over their secret keys for fear of, e.g., being slashed or losing their cryptocurrency assets.

But making cryptographic secrets external is only the start.
In Bitcoin, participation requires continuous use of hardware.
Without binding a hardware resource as PoW does, participation with a PKI is \emph{costless} and instantaneous with no restrictions on signing.
Consequently, even though the external adversary cannot make corrupt nodes sign messages once they are on break, it can ask corrupt nodes to ``pre-sign'' messages for the future before they take a break.
Honest nodes then cannot distinguish messages genuinely sent in the ongoing round from pre-signed ones.\footnote{Additionally, corrupt nodes may ``post-sign'' messages for past rounds after returning from a break.
Prior works have addressed this issue by ignoring past messages and relying only on messages from the latest round~\cite{MMR23,DNTT23}.
We inherit this approach.}

To distinguish \emph{present} messages from pre-signed ones, consider an ideal randomness beacon that emits a random value at each round, and a protocol in which each node signs the ongoing round's beacon value alongside its message, thereby binding each message to the round.
No node can predict future beacon values, so a corrupt node cannot pre-sign messages for rounds after it takes a break.
Thus, honest nodes can ignore messages that do not carry a signature on the current beacon value.
Unfortunately, realizing such an ideal beacon in fully fluctuating participation may be no easier than solving consensus itself, since nodes must agree on the randomness.
Instead,~\authcite{ENP25} implement a randomness beacon using VDFs.

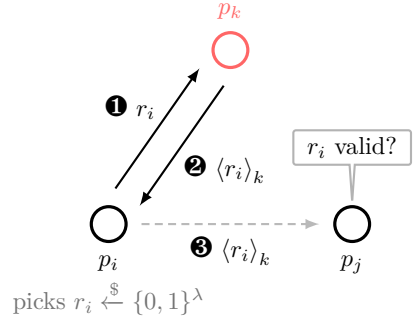
\begin{figure}[t]
    \centering
    \begin{tikzpicture}[
            scale=0.92,
            every node/.style={transform shape},
            msgarrow/.style={-latex, thick},
            fwdarrow/.style={-latex, thick, densely dashed, draw=gray!55},
        ]
        \def\xR{3.5}
        \def\yT{2.5}
        \def\xM{1.75}
        \def\diagOffX{1.6mm}
        \def\diagOffY{1.15mm}
        \def\diagLabelV{0mm}
        \def\diagLabelH{-1mm}
        \def\horizLabelGap{0.5mm}
        \def\nodeLabelGap{0.5mm}
        \def\nonceTextGap{4mm}
        \def\horizEndGap{1.8mm}
        \def\diagEndGap{1.8mm}
        \newcommand{\stepmark}[1]{{\Large\ding{#1}}}

        \node[MEMBER, MD] (hi) at (0,   0)   {};
        \node[MEMBER, MD] (hj) at (\xR, 0)   {};
        \node[MEMBER, MD, ADV] (a)  at (\xM, \yT) {};

        \node[below=\nodeLabelGap of hi, font=\normalsize] {$p_i$};
        \node[below=\nodeLabelGap of hj, font=\normalsize] {$p_j$};
        \node[above=\nodeLabelGap of a,  font=\normalsize, text=red!60] {$p_k$};

        \node[font=\normalsize, text=black!55, align=center,
            below=\nonceTextGap of hi]
        {picks $\nonce_i \xleftarrow{\$} \{0,1\}^\lambda$};

        \draw[msgarrow, transform canvas={shift={(-\diagOffX, \diagOffY)}}, shorten <=\diagEndGap, shorten >=\diagEndGap] (hi) --
        node[pos=0.5, above left=\diagLabelV and \diagLabelH, font=\normalsize]
        {\stepmark{182}~$\nonce_i$}
        (a);

        \draw[msgarrow, transform canvas={shift={(\diagOffX, -\diagOffY)}}, shorten <=\diagEndGap, shorten >=\diagEndGap] (a) --
        node[pos=0.5, below right=\diagLabelV and \diagLabelH, font=\normalsize]
        {\stepmark{183}~$\signed[k]{\nonce_i}$}
        (hi);

        \draw[fwdarrow, shorten <=\horizEndGap, shorten >=\horizEndGap] (hi.east) --
        node[below=\horizLabelGap, pos=0.5, font=\normalsize]
        {\stepmark{184}~$\signed[k]{\nonce_i}$}
        (hj.west);

        \node[LABELCALLOUT,
            callout absolute pointer={($(hj.north) + (0, 0.5mm)$)},
            draw=gray!60, font=\normalsize, fill=white,
            inner xsep=2mm, inner ysep=1mm]
        at ($(hj) + (0, 1.1cm)$)
        {$\nonce_i$ valid?};

    \end{tikzpicture}
    \caption[Challenge-response]{%
        Challenge-response between honest challenger~$p_i$, adversarial prover~$p_k$, and honest observer~$p_j$ (\honpic~honest; \advpic~adversarial).
        \ding{182}~$p_i$ samples a fresh random nonce~$\nonce_i \in \{0,1\}^\lambda$ (for security parameter~$\lambda$) and sends it to~$p_k$.
        \ding{183}~$p_k$ responds with a signature~$\signed[k]{\nonce_i}$; since $\nonce_i$ was unpredictable to $p_k$ in advance, a valid response proves that $p_k$ was participating at the time of the challenge.
        \ding{184}~$p_i$ may forward the proof~$\signed[k]{\nonce_i}$ to~$p_j$, but the proof is not \emph{transferable}: $p_j$ cannot verify that $\nonce_i$ was sampled freshly and at random by $p_i$, and so cannot conclude that $p_k$ was participating.%
    }
    \label{fig:visual-gw1}
\end{figure}

Our approach replaces the VDF with a simple protocol that achieves the same unpredictability \emph{interactively}, where each node samples its own fresh randomness instead of consulting a beacon.
A node $p_i$ samples a fresh random nonce and sends it to a node $p_k$ as a challenge; if $p_k$ returns a signature on the nonce (\cref{fig:visual-gw1}), $p_i$ concludes $p_k$ is currently participating since $p_k$ could not have pre-signed the response before it received the nonce.
However, this verification is not \emph{transferable}: $p_i$ has no way to prove $p_k$'s participation to a third node $p_j$.
So even when $p_i$ accepts $p_k$'s messages, $p_j$ may rightfully ignore them.
By contrast, the freshness of an ideal randomness beacon or VDF output is verifiable by any node.

\FloatBarrier
\begin{figure*}[t]
    \centering
    \captionsetup[subfigure]{skip=2pt}
    \def\xR{3.5}
    \def\yT{2.5}
    \def\xM{1.75}
    \def\diagOffX{1.6mm}
    \def\diagOffY{1.15mm}
    \def\horizLaneOff{1.6mm}
    \def\diagLabelV{-2mm}
    \def\diagLabelH{0mm}
    \def\horizLabelGap{0mm}
    \def\horizLowerLabelGap{0.5mm}
    \def\nodeLabelGap{0.5mm}
    \def\nonceTextGap{4mm}
    \def\horizEndGap{1.8mm}
    \def\diagEndGap{1.8mm}
    \newcommand{\stepmark}[1]{{\Large\ding{#1}}}
    \begin{subfigure}[b]{0.48\textwidth}
        \centering
        \begin{tikzpicture}[
                scale=0.86,
                every node/.style={transform shape},
                msgarrow/.style={-latex, thick},
                fwdarrow/.style={-latex, thick, densely dashed, draw=gray!55},
            ]
            \node[MEMBER, MD] (l-p2) at (0,   0)   {};
            \node[MEMBER, MD] (l-p3) at (\xR, 0)   {};
            \node[MEMBER, MD, ADV] (l-p1) at (\xM, \yT) {};
            \node[below=\nodeLabelGap of l-p3, font=\normalsize] {$p_j$};
            \node[below=\nodeLabelGap of l-p2, font=\normalsize] {$p_i$};
            \node[above=\nodeLabelGap of l-p1, font=\normalsize, text=red!60] {$p_k$};
            \node[font=\normalsize, text=black!55, align=center,
                below=\nonceTextGap of l-p2]
            {picks $r_i \xleftarrow{\$} \{0,1\}^\lambda$};
            \node[font=\normalsize, text=black!55, align=center,
                below=\nonceTextGap of l-p3]
            {picks $r_j \xleftarrow{\$} \{0,1\}^\lambda$};
            \draw[msgarrow, transform canvas={shift={(0, \horizLaneOff)}}, shorten <=\horizEndGap, shorten >=\horizEndGap] (l-p3.west) --
            node[above=\horizLabelGap, pos=0.5, font=\normalsize]
            {\stepmark{182}~$r_j$}
            (l-p2.east);
            \draw[msgarrow, transform canvas={shift={(-\diagOffX, \diagOffY)}}, shorten <=\diagEndGap, shorten >=\diagEndGap] (l-p2) --
            node[pos=0.53, above left=\diagLabelV and \diagLabelH, font=\normalsize]
            {\stepmark{183}~$r_i, r_j$}
            (l-p1);
            \draw[msgarrow, transform canvas={shift={(\diagOffX, -\diagOffY)}}, shorten <=\diagEndGap, shorten >=\diagEndGap] (l-p1) --
            node[pos=0.47, below right=\diagLabelV and \diagLabelH, font=\normalsize]
            {\stepmark{184}~$\signed[k]{r_i}$}
            (l-p2);
            \node[below=10mm of l-p2, font=\normalsize]
            {$W_0=\{p_k\}$};
            \node[below=10mm of l-p3, font=\normalsize]
            {$W_0 = \emptyset$};
        \end{tikzpicture}
        \caption{}
        \label{fig:visual-gw2a}
    \end{subfigure}%
    \hfill
    \begin{subfigure}[b]{0.48\textwidth}
        \centering
        \begin{tikzpicture}[
                scale=0.86,
                every node/.style={transform shape},
                msgarrow/.style={-latex, thick},
                fwdarrow/.style={-latex, thick, densely dashed, draw=gray!55},
            ]
            \node[MEMBER, MD] (r-p2) at (0,   0)   {};
            \node[MEMBER, MD] (r-p3) at (\xR, 0)   {};
            \node[MEMBER, MD, ADV] (r-p1) at (\xM, \yT) {};
            \node[below=\nodeLabelGap of r-p2, font=\normalsize] {$p_i$};
            \node[below=\nodeLabelGap of r-p3, font=\normalsize] {$p_j$};
            \node[above=\nodeLabelGap of r-p1, font=\normalsize, text=red!60] {$p_k$};
            \node[font=\normalsize, text=black!55, align=center,
                below=\nonceTextGap of r-p2]
            {picks $r_i \xleftarrow{\$} \{0,1\}^\lambda$};
            \node[font=\normalsize, text=black!55, align=center,
                below=\nonceTextGap of r-p3]
            {picks $r_j \xleftarrow{\$} \{0,1\}^\lambda$};
            \draw[msgarrow, transform canvas={shift={(0, \horizLaneOff)}}, shorten <=\horizEndGap, shorten >=\horizEndGap] (r-p3.west) --
            node[above=\horizLabelGap, pos=0.5, font=\normalsize]
            {\stepmark{182}~$r_j$}
            (r-p2.east);
            \draw[msgarrow, transform canvas={shift={(-\diagOffX, \diagOffY)}}, shorten <=\diagEndGap, shorten >=\diagEndGap] (r-p2) --
            node[pos=0.53, above left=\diagLabelV and \diagLabelH, font=\normalsize]
            {\stepmark{183}~$r_i, r_j$}
            (r-p1);
            \draw[msgarrow, transform canvas={shift={(\diagOffX, -\diagOffY)}}, shorten <=\diagEndGap, shorten >=\diagEndGap] (r-p1) --
            node[pos=0.47, below right=\diagLabelV and \diagLabelH, font=\normalsize]
            {\stepmark{184}~$\signed[k]{r_i}, \mathbf{\signed[k]{r_j}}$}
            (r-p2);
            \draw[fwdarrow, transform canvas={shift={(0, -\horizLaneOff)}}, shorten <=\horizEndGap, shorten >=\horizEndGap] (r-p2.east) --
            node[below=\horizLowerLabelGap, pos=0.5, font=\normalsize]
            {\stepmark{185}~$\signed[k]{r_j}$}
            (r-p3.west);
            \node[below=10mm of r-p2, font=\normalsize]
            {$W_1=\{p_k\}$};
            \node[below=10mm of r-p3, font=\normalsize]
            {$W_0=\{p_k\}$};
        \end{tikzpicture}
        \caption{}
        \label{fig:visual-gw2b}
    \end{subfigure}
    \caption[Two-grade challenge-response]{%
        Two scenarios of two-grade challenge-response among honest nodes~$p_i, p_j$ and adversarial node~$p_k$ (
        \honpic~honest; \advpic~adversarial).
        \ding{182}~$p_j$ samples~$r_j \in \{0,1\}^\lambda$ and sends it to~$p_i$.
        \ding{183}~$p_i$ sends~$r_i, r_j$ to~$p_k$.
        \ding{184}~In (\subref{fig:visual-gw2a}), $p_k$ signs only~$r_i$; in (\subref{fig:visual-gw2b}), $p_k$ signs both~$r_i$ and~$r_j$.
        \ding{185}~$p_i$ forwards what $p_k$ signed to~$p_j$, yielding different grade assignments in the two scenarios.
    }
    \label{fig:visual-gw2}
\end{figure*}
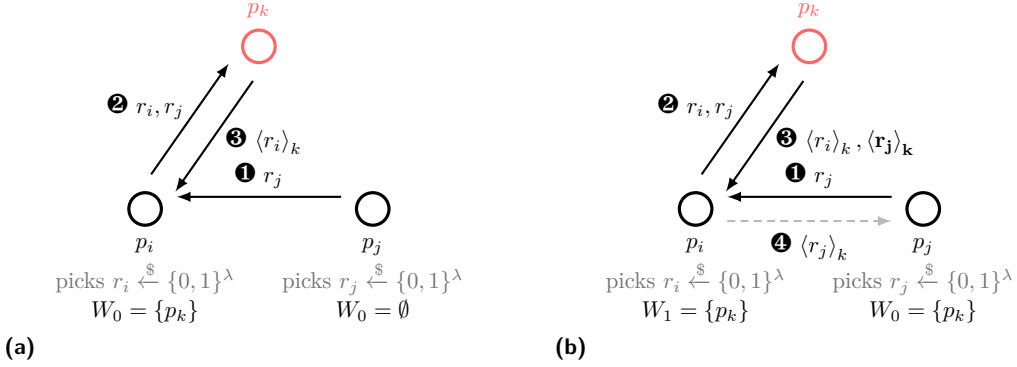

To add transferability to the verification of $p_k$'s participation, we additionally have node $p_j$ send a random nonce to $p_i$, which $p_i$ sends to $p_k$ along with its own.
In one scenario (\cref{fig:visual-gw2a}), $p_k$ only signs $p_i$'s nonce, in which case we are in the prior example.
However, $p_k$ may now also sign $p_j$'s nonce (\cref{fig:visual-gw2b}), and $p_i$ can meaningfully forward the signature to $p_j$.
In the first scenario, only $p_i$ knows that $p_k$ is participating, so we say $p_i$ considers $p_k$ participating with \emph{grade $0$}.
In the second, $p_i$ additionally holds evidence that it can \emph{transfer} to $p_j$, so $p_k$ has \emph{grade $1$} for $p_i$ and grade $0$ for $p_j$.

These grades resemble the hierarchy of distributed knowledge of~\authcite{HM90}: at grade $0$, a node $p_i$ knows $p_k$ is participating; at grade $1$, $p_i$ knows that $p_j$ knows it too, \emph{and}, therefore, that $p_j$ will accept any messages $p_i$ forwards from $p_k$.
However, even if $p_j$ accepts a forwarded message from $p_k$, it cannot distinguish whether or not $p_i$ forwarded the message to all honest nodes.
At grade $2$, then, $p_i$ knows that $p_j$ knows that all honest nodes will accept the message, allowing $p_j$ to forward the message itself.
With every additional grade, the message can be forwarded one further hop while remaining acceptable.
We call this primitive \emph{graded wakeness} (\cref{immobile:gw}).
It underpins our \emph{$2$-grade graded agreement} protocol (\cref{immobile:ga}, adapted from~\authcite{DSTZ24}), where each node's input must be forwarded for two rounds, requiring graded wakeness with three grades.
Chaining two instances of $2$-grade graded agreement yields $3$-grade graded agreement (\cref{immobile:ga3}), which we use to obtain consensus secure under fully fluctuating participation (\cref{immobile:ba,mobile}).

Note that our focus in this work is to establish feasibility: that PoW and VDFs are not required for consensus with fully fluctuating participation.
We leave improvements relevant to practical deployments, such as latency and message complexity, for future work.

\section{Model Preliminaries}
\label{model}

\subparagraph*{Nodes}
We now formalize the fully fluctuating participation model, which builds upon the model of~\authcite{ENP25}.
We operate in the permissioned setting with a predefined set $\nodes$ of $n$ nodes.
We assume a public key infrastructure (PKI), so each node's public key is known to all others.
The adversary is a probabilistic polynomial-time (PPT) algorithm that can exert control over three axes---communication, sleepiness, and corruption---subject to various constraints introduced below.

\subparagraph*{Communication}
Time proceeds in lock-step \emph{rounds}.
We assume a \emph{synchronous} network such that honest nodes at round $t$ receive all messages sent to them from any prior round $t' < t$.
(Note that consensus with sleepy nodes is \emph{impossible} in partial synchrony~\cite{GL02,NTT20a,LR23}.)
For simplicity, one can assume messages sent to asleep nodes are buffered until they wake up.
In practical systems, nodes do not need to buffer messages while asleep because they can retrieve relevant messages on demand from a peer-to-peer network after they wake up.
The adversary sees every message sent over the network and controls the timing of message delivery, subject to the synchrony constraint.

\subparagraph*{Sleepiness}
The adversary can put nodes to \emph{sleep} (i.e., make them take a break) in a \emph{mildly adaptive} way, where it may only put nodes to sleep at the start of a round $t$. %
This reflects the model's intention that nodes consciously go to sleep at appropriate round boundaries.
Sleeping mid-round without executing all protocol steps is corrupt behavior.
\emph{Asleep} nodes do not receive or send messages and do not execute the protocol.
We denote the set of awake nodes in round $t$ as $\awake_t \subseteq \nodes$.

\subparagraph*{Corruption}
The adversary can adaptively corrupt nodes during a round $t$.
We denote the set of nodes corrupt at any time in the round as $\adv_t$.
Nodes not corrupted by the adversary (called \emph{honest}) follow the protocol when they are awake.
We denote the set of nodes honest for the entire round as $\hon_t = \nodes \setminus \adv_t$.
In \cref{immobile}, we consider an \emph{immobile} adversary where corrupt nodes remain corrupt (i.e., $\adv_t \subseteq \adv_{t + 1}$ for all $t$).
In \cref{mobile}, we consider a \emph{mobile} adversary that can uncorrupt previously corrupt nodes.

\emph{Uncorruption} resets the node's protocol state to its initial state; the node is then asleep and upon waking behaves as if it woke up for the first time, receiving all messages sent to it in prior rounds.
This definition is made quite natural by the external adversary model (discussed below) since signing and VRF keys are never handed to the adversary.
Furthermore, the fully fluctuating sleepy model naturally resolves the question of post-uncorruption state since it already models a node waking up for the first time.
Resetting the uncorrupted node's state ensures the node completely rids itself of the adversary's presence, whether that be from malware, bribery, or software bugs.

\subparagraph*{Schedules}
A \emph{schedule} is defined as an infinite sequence of rounds of the adversary's choice of awake and corrupt nodes at each round.
For rounds $0 \leq t_1 \leq t_2$ in a schedule, we define the following sets of nodes across the interval:
\begin{itemize}
    \item $\hons{t_1}{t_2} = \bigcap_{t_1 \leq t \leq t_2} \hon_{t} \cap \awake_{t}$ --- nodes honest and awake at \emph{every} round in $[t_1, t_2]$;
    \item $\advs{t_1}{t_2} = \bigcup_{t_1 \leq t \leq t_2} \adv_{t} \cap \awake_{t}$ --- nodes corrupt and awake at \emph{some} round in $[t_1, t_2]$.
\end{itemize}
Note that $\hons{t_1}{t_2}$ defines honest nodes that do \emph{not} fluctuate in their participation over the interval.

We distinguish schedules of the adversary as \emph{admissible} in the $(\rec, \sbl, \rho)$-fully fluctuating sleepy model if for every round $t \geq 0$,
\[
    |\hons{t - \rec}{t + \sbl}| > \rho \cdot |\advs{t}{t + \sbl}|.
\]
For $t < 0$, let $\awake_t = \awake_0$ and $\hon_t = \hon_0$.
The two time parameters $\rec$ and $\sbl$ are non-negative integers that adjust the interval sizes of the honest and corrupt sets; $\rho$ is the minimum non-negative ratio allowed between these sets.
The parameter $\rec$ captures how many rounds an honest node must remain awake to \emph{recover} from being asleep.
The \emph{stable} period parameter $\sbl$ specifies how long a node must remain honest and awake to meaningfully participate and count toward the honest set.
A corrupt node counts toward the corrupt set if it is awake at any point during the stable period since it may pre-sign and post-sign messages to impersonate a stable node.
In \cref{def:stable-nodes}, we define two named honest sets that will be useful.
\begin{definition}[Stable Nodes]
    \label{def:stable-nodes}
    For a round $t$, an honest node $p$ is said to be \textbf{\stable[$t$]} if $p \in \hons{t}{t + \sbl}$, and \textbf{\vstable[$t$]} if $p \in \hons{t - \rec}{t + \sbl}$.
\end{definition}

\subparagraph*{Cryptographic Primitives and External Adversary}

Let $\lambda$ be a security parameter.
We make use of two cryptographic primitives: signatures and verifiable random functions.
We model them in an idealized fashion using oracles $\osig$ and $\ovrf$ that are defined as follows:
\begin{enumerate}
    \item Oracle $\osig$ accepts queries of the form $(p, m)$ from an \emph{awake} node $p$ with a message $m$.
          The oracle responds with signature $\signed[p]{m}$ in the same round as the query, where $\signed[p]{m}$ is shorthand for a signed message and is assumed to carry node $p$'s identity.
    \item Oracle $\ovrf$ accepts queries of the form $(p, m)$ from an \emph{awake} node $p$ and responds with $(\vrfo{m}, \pi)$.
          Here, $\sk$ is the secret key of node $p$ that is only known by $\ovrf$.
\end{enumerate}

In the external adversary model, nodes are not given secret keys to the cryptographic primitives used.
Instead, each node $p$ (honest or corrupt) can only issue queries to the oracles of the form $(p, \cdot)$ when awake.
Under the PKI, any node may verify signatures and VRF outputs using the known verification algorithms.

Throughout, we will say a function $\negl{\lambda}$ is \emph{negligible} if for all $c > 0$, there exists a $\lambda_0$ such that $\negl{\lambda} < \frac{1}{\lambda^{c}}$ for all $\lambda > \lambda_0$.

\subparagraph*{Byzantine Agreement}

Let $\valspace$ be some value space.
We define consensus in this work as the Byzantine agreement problem in \cref{def:agreement}.

\begin{definition}[Byzantine Agreement]
    \label{def:agreement}
    Each node has an input value in $\valspace$ at round $0$ and outputs a value in $\valspace$.
    \begin{itemize}
        \item \emph{Safety:} If an honest node outputs $x$ and an honest node outputs $x'$, then $x = x'$.
        \item \emph{Validity:} If each honest node awake at round $0$ inputs $x$, then no honest node outputs $x' \neq x$.
        \item \emph{Liveness:} There exists a round $t_0$ such that for every $t \geq t_0$, every \vstable[$t$] node outputs.
    \end{itemize}
\end{definition}

\section{Consensus in the Fully Fluctuating Sleepy Model}

\label{immobile}

In this section, we give our base agreement protocol in the original external adversary model of~\authcite{ENP25} under fully fluctuating participation.
As such, the adversary in this section cannot uncorrupt nodes.
In \cref{mobile}, we extend the protocol to handle uncorruption.
The protocol is largely the same; uncorruption comes with different liveness guarantees, so we begin with a complete analysis without uncorruption that captures the core ideas.

Our Byzantine agreement ($\ba$) protocol proceeds in sequential \emph{views} with each view calling an instance of our $3$-grade graded agreement protocol $\ga3$ (\cref{immobile:ga3}).
Our $3$-grade graded agreement protocol is built from two instances of our $2$-grade graded agreement protocol $\ga2$ (\cref{immobile:ga}).
The $\ga2$ protocol is based on~\authcite{DSTZ24}'s construction,  which we upgrade to tolerate fully fluctuating participation through the use of our graded wakeness protocol $\gw3$ (\cref{immobile:gw}).
At the end of every view of $\ba$, stable nodes participate in a VRF-based leader election to determine the input for the next view.
In this work, we focus on single-shot agreement, but the structure closely resembles the total-order broadcast protocol of~\authcite{DSTZ24}.

\subsection{Graded Wakeness}
\label{immobile:gw}

Recall from \cref{intro:technical} that graded wakeness allows nodes to filter for messages from awake nodes \emph{and} ensure that other awake nodes accept such messages when forwarded.
We represent this knowledge as \emph{graded wakeness sets}, where nodes in the grade $1$ set $\gawake[1]^{i}$ are included in $\gawake[0]^{j}$ for any stable node $p_j$.
In the example of \cref{fig:visual-gw2b}, node $p_i$, with $\gawake[1]^{i} = \{p_k\}$, knows that $p_j$, with $\gawake[0]^{j} = \{p_k\}$, will accept messages originating from $p_k$.
In \cref{immobile:ga}, we show that three grades of graded wakeness suffice for the graded agreement protocol.

\Cref{def:graded-wakeness} formally defines the desired properties of these sets, where the property just described is \emph{graded delivery}.
The last two properties naturally follow from what we expect from a wakeness set: that nodes asleep at the start of the protocol are excluded (\emph{asleep exclusion}) and that stable nodes are included in all the sets (\emph{honest inclusion}).
Additionally, honest inclusion is ``all-or-nothing,'' such that a node honest since $t_s$ is either in \emph{all} of a \stable node's graded wakeness sets or in none of them.
\begin{definition}[Graded Wakeness]
    \label{def:graded-wakeness}
    Parameterized by a start round $t_s$, each \stable node $p_i$ outputs $\gawake[g]^{i} \subseteq \nodes$ for each grade $g \in \{0, ..., \gmax - 1\}$ with the following properties, except with negligible probability:
    \begin{itemize}
        \item \emph{Graded Delivery:} For any \stable node $p_j$ and every $g \geq 1$, $\gawake[g]^i \subseteq \gawake[g-1]^j$.
        \item \emph{Asleep Exclusion:} $\gawake[0]^i$ includes only nodes awake after $t_s$.
        \item \emph{Honest Inclusion:} $\gawake[\gmax - 1]^i$ includes all \stable nodes and every node in $\gawake[0]^i$ that has remained honest since $t_s$.
    \end{itemize}
\end{definition}
Note the unconventional notion that the protocol may start at any round $t_s \geq 0$.
This captures the adversary's ability to pre-sign messages \emph{before} the protocol starts during long-running executions (such as during agreement).

\begin{figure}[tb]
  \centering
  \begin{tcolorbox}
    A node $p_i$ executes the below steps starting from a round $t_s$ and, at every round, forwards every received $\pingmsgtype$ and $\ackmsgtype$ message to all nodes.
    Let $t$ denote the rounds since $t_s$, where $t_s + t$ is the absolute round in the entire execution.
    During execution, node $p_i$ stores $\nonces^{t}$ as the set of nonces received in $\pingmsgtype$ messages by round $t$, including its own.
    A node $p_j$ has \emph{acknowledged} a nonce $\nonce$ once $p_i$ receives $\signedackmsg[j]{\nonce}$.
    Only a node awake since round $t = 0$ executes the steps at rounds $t \geq 3$.
    \begin{protocolsteps}
      \item[$(t = 0):$] %
            Sample nonce $\nonces^{0} \overset{\$}{\gets} \{0, 1\}^{\lambda}$.
            Multicast $\signedpingmsg[i]{\nonces^{0}}$.
      \item[$(t = 3):$] %
            Multicast $\signedackmsg[i]{\nonce}$ for every nonce $\nonce \in \nonces^{3}$.
      \item[$(t = 4):$] %
            Output $\gawake[2]$ as the nodes that acknowledged every nonce in $\nonces^{2}$.
      \item[$(t = 5):$] %
            Output $\gawake[1]$ as the nodes that acknowledged every nonce in $\nonces^{1}$.
      \item[$(t = 6):$] %
            Output $\gawake[0]$ as the nodes that acknowledged $\nonces^{0}$.
    \end{protocolsteps}
  \end{tcolorbox}
  \caption{Graded wakeness with $3$ grades for a node $p_i$.}
  \label{fig:gw3}
\end{figure}

We detail the graded wakeness protocol ($\gw3$) for $\gmax = 3$ grades in \cref{fig:gw3}.
For simplicity, we reference the rounds of the protocol relative to its start time when it is clear from context, so round $0$ of the protocol maps to round $t_s$ in the entire execution.

Nodes start by sampling a \emph{nonce} challenge and multicasting it alongside a signature from the oracle $\osig$ in a $\pingmsgtype$ message.
The nonces are then forwarded for two rounds; in each round, nodes store the nonces received so far.
At round $3$, nodes \emph{acknowledge} every received nonce by signing and multicasting it in an $\ackmsgtype$ message.
From round $4$ on, nodes begin to output the graded wakeness sets from highest to lowest.
For grade $2$ and a node $p_i$, if a node $p_j$ has acknowledged all the nonces $p_i$ received by round $2$, then $p_j$ is included in $\gawake[2]^{i}$.
Node $p_i$ then forwards the acknowledgments to all nodes.
This continues until round $6$, once nodes have output graded wakeness sets for every grade.
As a result, graded delivery easily follows from the implication that nodes that acknowledge $p_i$'s round $2$ nonces also acknowledge all round $1$ nonces of all other nodes that forward them to $p_i$.
We now prove it precisely along with the other properties of \cref{def:graded-wakeness}.

\begin{lemma}[Graded Delivery]
    \label{lem:immobile-gw-gd}
    For any \stable nodes $p_i$ and $p_j$, $\gawake[2]^{i} \subseteq \gawake[1]^{j}$ and $\gawake[1]^{i} \subseteq \gawake[0]^{j}$.
\end{lemma}

\begin{proof}
    Since $p_i, p_j \in \hons{t_s}{t_s + 3}$, all nonces are forwarded from rounds $0$ to $2$, so
    $\nonces^{0, j} \subseteq \nonces^{1, i}$ and $\nonces^{1, j} \subseteq \nonces^{2, i}$.
    All nodes in $\gawake[2]^{i}$ acknowledged all the nonces in $\nonces^{2, i}$.
    Node $p_i$ receives these acknowledgments by round $4$, and they are forwarded such that node $p_j$ receives them by round $5$.
    Since $\nonces^{1, j} \subseteq \nonces^{2, i}$, all nodes in $\gawake[2]^{i}$ acknowledged all the nonces in $\nonces^{1, j}$, so $\gawake[2]^{i} \subseteq \gawake[1]^{j}$.

    Similarly, all nodes in $\gawake[1]^{i}$ acknowledged all the nonces in $\nonces^{1, i}$.
    Node $p_i$ receives these acknowledgments by round $5$, and they are forwarded such that node $p_j$ receives them by round $6$.
    Since $\nonces^{0, j} \subseteq \nonces^{1, i}$, all nodes in $\gawake[1]^{i}$ acknowledged all the nonces in $\nonces^{0, j}$, so $\gawake[1]^{i} \subseteq \gawake[0]^{j}$.
\end{proof}

\begin{lemma}[Asleep Exclusion]
    \label{lem:immobile-gw-exclusion}
    For any \stable node $p_i$, $\gawake[0]^{i}$ includes only nodes awake after $t_s$.
\end{lemma}

\begin{proof}
    Note that any node $p_j$ that is not awake after $t_s$ is not awake during the execution of the protocol.
    Such a node does not receive the nonce $\nonces^{0, i}$.
    Therefore, except with probability $\poly(n)/2^{\lambda}$,
    it cannot acknowledge $\nonces^{0, i}$, so $p_j \not\in \gawake[0]^i$.
\end{proof}

\begin{lemma}[Honest Inclusion]
    \label{lem:immobile-gw-inclusion}
    For any \stable node $p_i$, $\gawake[2]^{i}$ includes all \stable nodes and every node in $\gawake[0]^{i}$ that has remained honest since $t_s$, i.e., $\hons{t_s}{t_s + 6} \cup (\gawake[0]^{i} \setminus \advs{t_s}{t_s + 6}) \subseteq \gawake[2]^i$.
\end{lemma}

\begin{proof}
    Let $p_j$ be a \stable node that receives $p_i$'s nonce set $\nonces^{2, i}$ by round $3$.
    At round $3$, node $p_j$ acknowledges every nonce in $\nonces^{2, i}$, and the acknowledgments are forwarded such that $p_i$ receives them by round $4$.
    Therefore, $p_j \in \gawake[2]^i$, so $\hons{t_s}{t_s + 6} \subseteq \gawake[2]^i$.

    For a node $p_j \in \gawake[0]^i \setminus \advs{t_s}{t_s + 6}$, $p_j$ acknowledged $p_i$'s nonce $\nonces^{0, i}$ at round $3$.
    Node $p_j$ must have also received $p_i$'s nonces $\nonces^{2, i}$ and acknowledged them all, which $p_i$ received by round $4$.
    Thus, in round $4$, $p_i$ must have included $p_j$ in $\gawake[2]^i$, and we have that $\hons{t_s}{t_s + 6} \cup (\gawake[0]^{i} \setminus \advs{t_s}{t_s + 6}) \subseteq \gawake[2]^i$.
\end{proof}

\begin{theorem}
    \label{thm:immobile-gw3}
    \Cref{fig:gw3} implements graded wakeness with start round $t_s \geq 0$ and $\gmax = 3$ for schedules admissible in the fully fluctuating sleepy model with $\sbl \geq 6$.
\end{theorem}
\begin{proof}%
    With $\sbl \geq 6$, every \stable node $p_i$ is in $\hons{t_s}{t_s + 6}$.
    Therefore, every \stable node outputs $\gawake[g]^i$ for each grade $g$.
    We then have \emph{graded delivery} by \cref{lem:immobile-gw-gd}, \emph{asleep exclusion} by \cref{lem:immobile-gw-exclusion}, and \emph{honest inclusion} by \cref{lem:immobile-gw-inclusion}.
\end{proof}

\subsection{Graded Agreement}
\label{immobile:ga}

We now outline how we adapt the graded agreement protocol from~\authcite{DSTZ24} with $\gmax = 2$ to tolerate fully fluctuating participation.
\Cref{def:graded-agreement} gives our definition of graded agreement.
\begin{definition}[Graded Agreement]
    \label{def:graded-agreement}
    Parameterized by a start round $t_s$, each node has an input value in $\valspace \cup \{\noinput, \bot\}$, where $\bot$ indicates a default value, and $\noinput$ indicates the node has no input.
    Each node outputs a single pair $(x, g)$ where $x \in \valspace \cup \{\bot\}$ and $g \in \{0, ..., \gmax - 1\}$.
    \begin{itemize}
        \item \emph{Graded Delivery:} For any $g > 0$, if a \stable node outputs $(x, g)$, then every \stable node outputs $(x, g')$ for $g - 1 \leq g' \leq g + 1$.
        \item \emph{Consistency:} If an honest node outputs $(x, g)$ for $g > 0$, then no honest node outputs $(x', \cdot)$ for $x' \neq x$.
        \item \emph{Validity:} If every honest and awake node at round $t_s$ with an input in $\valspace$ inputs the same $x$, then every \stable node outputs $(x, \gmax - 1)$.
        \item \emph{Integrity:} If no honest and awake node at round $t_s$ inputs $x \in \valspace$, then no honest node outputs $(x, \cdot)$.
    \end{itemize}
\end{definition}
Note the unconventional notion that honest nodes may have no input ($\noinput$).
This is used in our agreement protocol for nodes in \emph{recovery} that have no input even though they are honest and awake at the start of the latest graded agreement.
To get up to date, such nodes only observe the output of graded agreement.

First, recall the \emph{time-shifted quorum} technique of~\authcite{MR22} that underlies the graded agreement protocol of~\authcite{DSTZ24}.
Traditional agreement protocols rely on quorums of more than $n/2$ nodes, which any node can verify and forward, since no two such quorums can disagree.
In the sleepy model, a node may never see that many nodes if a majority remain asleep, so it can only weigh a value against the inputs it actually receives.

Suppose, then, that a node accepts a value if more than half of the inputs it has received are for that value.
Consider three disjoint sets of awake nodes $P_1, P_2, P_3$ of equal size at round $t$, where the number of awake nodes $n_t = |\awake_{t}|$ is far less than $n$.
If the nodes in $P_3$ are corrupt and send their inputs only to $P_2$, then $P_1$ receives $n_t/3$ inputs while $P_2$ receives $2n_t/3$.
A value carried by just over $n_t/6$ inputs then suffices for $P_1$ but not for $P_2$, so these \emph{subjective} quorums do not transfer.
As with the challenge-response behind graded wakeness, the time-shifted quorum technique restores transferability.

After nodes send their input at the start of the protocol ($t=0$), the time-shifted quorum technique proceeds as follows for a node $p_i$:
\begin{itemize}
    \item For rounds $t = 1, 2, 3$, obtain the \emph{inputters} $\inputters^{t,i}$ as the set of nodes from which $p_i$ receives inputs by round $t$.
    \item For rounds $t = 2, 3$, obtain the \emph{equivocators} $\equivs^{t, i}$ as the set of nodes from which $p_i$ receives two distinct inputs by round $t$.
\end{itemize}
Throughout, nodes forward any received inputs to all nodes: for another node $p_j$, any inputs received by $p_i$ by round $t$ are guaranteed to be received by $p_j$ by the next round.
Therefore, $p_j$'s set of non-equivocating inputters at round $2$ includes all of $p_i$'s non-equivocating inputters from rounds $1$ to $3$, i.e., $\inputters^{1, i} \setminus \equivs^{3, i} \subseteq \inputters^{2, j} \setminus \equivs^{2, j}$.
This follows because $p_i$ forwards its inputs at round $1$, and $p_j$ forwards any evidence of equivocation among those inputters in round $2$.
Let $\inputters^{t, i}_{x}$ be the set of inputters in $\inputters^{t, i}$ that sent input $x$; then if $p_i$ has a value $x$ such that $|\inputters^{1, i}_{x} \setminus \equivs^{3, i}| > |\inputters^{3, i}|/2$, it outputs $x$ with grade $1$, and $p_j$ then outputs $x$ with at least grade $0$ since $|\inputters^{2, j}_{x} \setminus \equivs^{2, j}| > |\inputters^{2, j}|/2$.

Upgrading the protocol to tolerate fully fluctuating participation then amounts to obtaining graded wakeness sets for every round of forwarding, which would be $3$ graded wakeness sets for the above time-shifted quorum.
The above output invariants now become: $p_i$ outputs $x$ with grade $1$ if $|(\inputters^{1, i}_{x} \setminus \equivs^{3, i}) \cap \gawake[2]^{i}| > |\inputters^{3, i} \cap \gawake[0]^{i}|/2$, and $p_j$ then outputs $x$ with grade $0$ if $|(\inputters^{2, j}_{x} \setminus \equivs^{2, j}) \cap \gawake[1]^{j}| > |\inputters^{2, j} \cap \gawake[1]^{j}|/2$.
We can check that since $\gawake[2]^{i} \subseteq \gawake[1]^{j} \subseteq \gawake[0]^{i}$, we still have $(\inputters^{1, i}_{x} \setminus \equivs^{3, i}) \cap \gawake[2]^{i} \subseteq (\inputters^{2, j}_{x} \setminus \equivs^{2, j}) \cap \gawake[1]^{j}$ and $\inputters^{2, j} \cap \gawake[1]^{j} \subseteq \inputters^{3, i} \cap \gawake[0]^{i}$.

\begin{figure}[t]
  \centering
  \begin{tcolorbox}
    A node $p_i$ with input $x$ executes the below steps starting from a round $t_s$ and, at every round, forwards every received $\inputmsgtype$ message to all nodes.
    Let $t$ denote the rounds since $t_s$, where $t_s + t$ is the absolute round in the entire execution.
    During execution, node $p_i$ stores the sets $\inputters^{t}, \inputters^{t}_{y},$ and $\equivs^{t}$ at every round $t$ where, for nodes $p_j \in \nodes$:
    \begin{itemize}
      \item \emph{Inputters:} $p_j \in \inputters^{t}$ if $p_i$ has received $\signedinputmsg[j]{x}$ for any $x \in \valspace$ by round $t$.
      \item \emph{$y$-Inputters:} $p_j \in \inputters^{t}_y$ if $p_i$ has received $\signedinputmsg[j]{y}$ by round $t$.
      \item \emph{Equivocators:} $p_j \in \equivs^{t}$ if $p_i$ has received $\signedinputmsg[j]{x}$ and $\signedinputmsg[j]{y}$ for $x \neq y$ by round $t$.
    \end{itemize}
    \begin{protocolsteps}
      \item[$(t = 0):$]
      If $x \neq \noinput$: multicast $\signedinputmsg[i]{x}$.
      \item[$(t = 0, 1, 2, 3):$]
      Execute $\gw3$.
      \item[$(t = 4):$]
      Store $\inputters^{4}$. %
      \item[$(t = 5):$]
      Store $\inputters^{5}$ and $\equivs^{5}$. %
      \item[$(t = 6):$] If awake since round $0$:\leavevmode
      \begin{itemize}
        \item If $\exists y \in \valspace$ such that $|(\inputters^{4}_{y} \setminus \equivs^{6}) \cap \gawake[2]| > |\inputters^{6} \cap \gawake[0]|/2$: output $(y, 1)$.
        \item Else if $\exists y \in \valspace$ such that $|(\inputters^{5}_{y} \setminus \equivs^{5}) \cap \gawake[1]| > |\inputters^{5} \cap \gawake[1]|/2$: output $(y, 0)$.
        \item Otherwise, output $(\bot, 0)$.
      \end{itemize}
    \end{protocolsteps}
  \end{tcolorbox}
  \caption{Graded agreement with $2$ grades for a node $p_i$.}
  \label{fig:ga2}
\end{figure}

We give our $2$-grade graded agreement protocol in \cref{fig:ga2}.
As in \cref{immobile:gw}, we reference the rounds of the protocol relative to its start time $t_s$.
Note that the time-shifted quorum just described maps exactly to rounds $4$ to $6$.
We now prove that the protocol satisfies the properties of \cref{def:graded-agreement} for $\sbl \geq 6$ and $\rho \geq 1$ in the fully fluctuating sleepy model.
Furthermore, we assume \emph{very-stable input} as defined in \cref{def:immobile-ga-vstable-input}.
\begin{definition}[Very-Stable Input]
    \label{def:immobile-ga-vstable-input}
    For a start round $t_s$, a graded agreement protocol is said to have very-stable input if every \vstable node has input $x \neq \noinput$.
\end{definition}

\begin{lemma}[Graded Delivery]
    \label{lem:immobile-ga-gd}
    If a \stable node $p_i$ outputs $(x, 1)$, then every \stable node $p_j$ outputs $(x, g')$ for $g' \leq 1$.
\end{lemma}

    \begin{proof}%
        Since $p_i$ outputs $(x, 1)$, we have $|(\inputters^{4, i}_{x} \setminus \equivs^{6, i}) \cap \gawake[2]^{i}| > |\inputters^{6, i} \cap \gawake[0]^{i}|/2$, and $p_i$ has been awake since round $0$.
        Thus, $p_i$ forwards all $\inputmsgtype$ messages in $\inputters^{4, i}$ at round $4$, so $p_j$ receives them by round $5$, giving $\inputters^{4, i} \subseteq \inputters^{5, j}$.
        For any sender in $\inputters^{4, i}$ that $p_j$ observes equivocating by round $5$, $p_j$ forwards the equivocation evidence by round $5$, and $p_i$ receives it by round $6$, so $\equivs^{5, j} \subseteq \equivs^{6, i}$.
        Therefore, $\inputters^{4, i}_{x} \setminus \equivs^{6, i} \subseteq \inputters^{5, j}_{x} \setminus \equivs^{5, j}$.
        By graded delivery of $\gw3$, we have $\gawake[2]^{i} \subseteq \gawake[1]^{j}$.
        Combined, we get
        \begin{equation}
            \label{pf:immobile-ga-gd:vals}
            (\inputters^{4, i}_{x} \setminus \equivs^{6, i}) \cap \gawake[2]^{i} \subseteq (\inputters^{5, j}_{x} \setminus \equivs^{5, j}) \cap \gawake[1]^{j}.
        \end{equation}
        Similarly, $p_j$ forwards all values in $\inputters^{5, j}$ such that $p_i$ receives them by round $6$, i.e., $\inputters^{5, j} \subseteq \inputters^{6, i}$, and, by graded delivery of $\gw3$, we have $\gawake[1]^{j} \subseteq \gawake[0]^{i}$.
        Combined, we get
        \begin{equation}
            \label{pf:immobile-ga-gd:ins}
            \inputters^{5, j} \cap \gawake[1]^{j} \subseteq \inputters^{6, i} \cap \gawake[0]^{i}.
        \end{equation}
        Therefore, by \cref{pf:immobile-ga-gd:vals,pf:immobile-ga-gd:ins},
        $|(\inputters^{4, i}_{x} \setminus \equivs^{6, i}) \cap \gawake[2]^{i}| > |\inputters^{6, i} \cap \gawake[0]^{i}|/2 \geq |\inputters^{5, j} \cap \gawake[1]^{j}|/2$,
        so the grade $0$ condition holds at $p_j$ for $x$.
        Applying the same argument as above with $p_j$ in place of $p_i$ shows that if $p_j$ outputs with grade $1$, then its value meets $p_j$'s grade $0$ condition as well.
        Since the grade $0$ condition requires more than half of $|\inputters^{5, j} \cap \gawake[1]^{j}|$, only one value meets it, so node $p_j$ outputs $(x, g')$ for $g' \leq 1$.
    \end{proof}

\begin{lemma}[Consistency]
    \label{lem:immobile-ga-consistency}
    If an honest node $p_i$ outputs $(x, 1)$, then no honest node $p_j$ outputs $(x', \cdot)$ for $x' \neq x$.
\end{lemma}

    \begin{proof}%
        Nodes $p_i$ and $p_j$ will only output if they have been awake since the start of the protocol, so both lie in $\hons{t_s}{t_s + 6}$.
        Thus, by graded delivery (\cref{lem:immobile-ga-gd}), $p_j$ will not output $(x', \cdot)$ for $x' \neq x$.
    \end{proof}

For the following proofs of validity and integrity, we use \cref{lem:immobile-ins-majority}.
Let $\ins$ be the set of honest and awake nodes at round $t_s$ whose input is not $\noinput$.

\begin{lemma}[$\ins$-Majority]
    \label{lem:immobile-ins-majority}
    Assuming very-stable input (\cref{def:immobile-ga-vstable-input}), every \stable node $p_i$ has
    \[
        |(\ins \cap \gawake[2]^{i}) \setminus \advs{t_s}{t_s + 6}| > |\inputters^{6, i} \cap \gawake[0]^{i}|/2.
    \]
\end{lemma}

\begin{proof}
    We first bound the number of inputs $p_i$ receives in $\inputters^{6, i} \cap \gawake[0]^{i}$.
    Nodes in $\ins$ are honest and awake at round $0$, so they send an input, which $p_i$ receives by round $6$ in $\inputters^{6, i}$.
    By asleep exclusion of $\gw3$, every node in $\gawake[0]^{i}$ is awake after $t_s$, so any other node in $\gawake[0]^{i}$ that could send an input is corrupt and awake during the protocol; so $(\ins \cap \gawake[0]^{i}) \cup \advs{t_s}{t_s + 6}$ includes all nodes in $\gawake[0]^{i}$ that could send an input by round $6$, giving us:
    \begin{equation}
        \label{eq:immobile-ins-majority:upper}
        |\inputters^{6, i} \cap \gawake[0]^{i}| \leq |(\ins \cap \gawake[0]^{i}) \setminus \advs{t_s}{t_s + 6}| + |\advs{t_s}{t_s + 6}|,
    \end{equation}
    where the left summand indicates the set of honest inputs and the right indicates the set of corrupt inputs.

    We now lower bound the number of inputs in $(\ins \cap \gawake[2]^{i}) \setminus \advs{t_s}{t_s + 6}$.
    Under the assumption of very-stable input (\cref{def:immobile-ga-vstable-input}), we have that $\ins$ includes all \vstable nodes, and by honest inclusion of $\gw3$, we have that $\gawake[2]^{i}$ includes all \stable nodes.
    Further, by graded delivery of $\gw3$, we have $\gawake[2]^{i} \subseteq \gawake[0]^{i}$, and by honest inclusion, we have $\gawake[0]^{i} \setminus \advs{t_s}{t_s + 6} \subseteq \gawake[2]^{i}$; therefore, it follows that $\gawake[2]^{i}$ and $\gawake[0]^{i}$ include the same set of honest nodes, i.e., $\gawake[2]^{i} \setminus \advs{t_s}{t_s + 6} = \gawake[0]^{i} \setminus \advs{t_s}{t_s + 6}$.
    Therefore, by the assumption that we are in a schedule admissible in the fully fluctuating sleepy model, we have
    \[
        |\ins \cap (\gawake[0]^{i} \setminus \advs{t_s}{t_s + 6})|                               \geq |\hons{t_s - \rec}{t_s + \sbl}|                                                       > |\advs{t_s}{t_s + 6}|,
    \]
    since $\rho \geq 1$ and $\sbl \geq 6$.
    Finally, by \cref{eq:immobile-ins-majority:upper}, it follows that $|\ins \cap (\gawake[2]^{i} \setminus \advs{t_s}{t_s + 6})| > |\inputters^{6, i} \cap \gawake[0]^{i}| / 2$.
\end{proof}

\begin{lemma}[Validity]
    \label{lem:immobile-ga-validity}
    If every node in $\ins$ inputs $x$, then every \stable node $p_i$ outputs $(x, 1)$.
\end{lemma}

\begin{proof}
    Nodes in $\ins$ are honest and awake at round $0$, so they send $x$, which $p_i$ receives by round $4$ in $\inputters^{4, i}_{x}$.
    For the nodes in $\gawake[2]^{i}$ that remain honest until round $6$, they do not equivocate; by \cref{lem:immobile-ins-majority} ($\ins$-majority), we get
    $|(\inputters^{4, i}_{x} \setminus \equivs^{6, i}) \cap \gawake[2]^{i}| \geq |\ins \cap (\gawake[2]^{i} \setminus \advs{t_s}{t_s + 6})| > |\inputters^{6, i} \cap \gawake[0]^{i}| / 2$,
    and it follows that $p_i$ outputs $(x, 1)$.
\end{proof}

\begin{lemma}[Integrity]
    \label{lem:immobile-ga-integrity}
    If no node in $\ins$ inputs $x \in \valspace$, then no honest node $p_i$ outputs $(x, g)$ for every $g \in \{0, 1\}$.
\end{lemma}

\deferredproof{ga2proofs}{lem:immobile-ga-integrity}
\begin{defer}{ga2proofs}
    \begin{proof}[Proof (of \cref{lem:immobile-ga-integrity})]
        Note that for $p_i$ to output, it must be a stable node.
        For the sake of contradiction, assume node $p_i$ outputs $(x, 1)$, implying it received more than $|\inputters^{6, i} \cap \gawake[0]^{i}|/2$ inputs of $x$, i.e., $|(\inputters^{4, i}_{x} \setminus \equivs^{6, i}) \cap \gawake[2]^{i}| > |\inputters^{6, i} \cap \gawake[0]^{i}|/2$.
        By \cref{lem:immobile-ins-majority} ($\ins$-majority), we have
        $|\ins \cap (\gawake[2]^{i} \setminus \advs{t_s}{t_s + 6})| > |\inputters^{6, i} \cap \gawake[0]^{i}| / 2$.
        So, for $p_i$ to output $(x, 1)$, it must have received $x$ from a node in $\ins \cap (\gawake[2]^{i} \setminus \advs{t_s}{t_s + 6})$.
        However, no node in $\ins$ inputs $x$, so such a node would have to equivocate to send $x$ if it is corrupt after round $0$, contradicting that the input is in $\inputters^{4, i}_{x} \setminus \equivs^{6, i}$.
        Thus, no honest node outputs $(x, 1)$.
        The proof is identical for grade $0$, so we omit it.
    \end{proof}
\end{defer}

\begin{theorem}
    \label{thm:immobile-ga2}
    Let $t_s \geq 0$ be the start round.
    If every \vstable node has input $x \neq \noinput$ (very-stable input, \cref{def:immobile-ga-vstable-input}), then the protocol in \cref{fig:ga2} implements graded agreement with start round $t_s$ and $\gmax = 2$ for schedules admissible in the fully fluctuating sleepy model with $\sbl \geq 6$ and $\rho \geq 1$.
\end{theorem}

\begin{proof}
    Since $\sbl \geq 6$, each \stable node $p_i \in \hons{t_s}{t_s + 6}$, so we have \emph{graded delivery} by \cref{lem:immobile-ga-gd}, \emph{consistency} by \cref{lem:immobile-ga-consistency}, \emph{validity} by \cref{lem:immobile-ga-validity}, and \emph{integrity} by \cref{lem:immobile-ga-integrity}.
\end{proof}

\begin{ifnotdeferred}{ga3proofs}
    \subsubsection{Graded Agreement with 3 Grades}
\end{ifnotdeferred}
\begin{defer}{ga3proofs}
    \label{immobile:ga3}

    We build $3$-grade graded agreement from two sequential instances of $2$-grade graded agreement, $\ga2_0$ and $\ga2_1$, as shown in \cref{fig:immobile-ga3}.
    A node without an output from $\ga2_0$ inputs $\noinput$ into $\ga2_1$.
    We prove that the protocol satisfies the properties of \cref{def:graded-agreement} with $\gmax = 3$ for $\rec \geq 6$, $\sbl \geq 12$, and $\rho \geq 1$ in the fully fluctuating sleepy model, again assuming very-stable input (\cref{def:immobile-ga-vstable-input}).
    Since $\rec \geq 6$, every \vstable node for $\ga2_1$'s start round $6$ has been awake since round $0$; it therefore has an output $y_0 \neq \noinput$ from $\ga2_0$ that it inputs into $\ga2_1$, so $\ga2_1$ also has very-stable input.

    \begin{figure}[t]
        \centering
        \begin{tcolorbox}
            A node $p_i$ with input $x$ executes the below steps starting from a round $t_s$.
            Let $t$ denote the rounds since $t_s$, where $t_s + t$ is the absolute round in the entire execution.
            \begin{protocolsteps}
                \item[$(t = 0 ... 6):$] Execute $\ga2_{0}$ with input $x$.
                \item[$(t = 6):$] Let $(y_0, g_0) \gets $ the output of $\ga2_{0}$ if it exists; otherwise, let $(y_0, g_0) \gets (\noinput, 0)$.
                \item[$(t = 6 ... 12):$] Execute $\ga2_{1}$ with input $y_0$.
                \item[$(t = 12):$] If awake since round $0$:\leavevmode
                \begin{itemize}
                    \item Let $(y_1, g_1) \gets $ the output of $\ga2_{1}$.
                    \item If $g_0 = 1$: output $(y_0, 2)$.
                    \item Otherwise, output $(y_1, g_1)$.
                \end{itemize}
            \end{protocolsteps}
        \end{tcolorbox}
        \caption{Graded agreement with $3$ grades for a node $p_i$.}
        \label{fig:immobile-ga3}
    \end{figure}

    \begin{lemma}[Graded Delivery]
        \label{lem:immobile-ga3-gd}
        For any $g > 0$, if a \stable node $p_i$ outputs $(x, g)$, then every \stable node $p_j$ outputs $(x, g')$ for $g - 1 \leq g' \leq g + 1$.
    \end{lemma}

    \begin{proof}
        Note that a node only outputs if it has been awake since round $0$, so it is \stable for both $\ga2_0$ and $\ga2_1$.
        First, suppose $g = 2$, so that $p_i$ outputs $(x, 1)$ from $\ga2_0$.
        By graded delivery of $\ga2_0$, every \stable node $p_j$ outputs $(x, g_0)$ from $\ga2_0$ for $g_0 \leq 1$, and, by consistency of $\ga2_0$, no honest node outputs a value other than $x$ from $\ga2_0$.
        Every honest and awake node at round $6$ with an input $\neq \noinput$ therefore inputs $x$ into $\ga2_1$, so by validity of $\ga2_1$, node $p_j$ outputs $(x, 1)$ from $\ga2_1$.
        Node $p_j$ then outputs $(x, 2)$ if $g_0 = 1$ and $(x, 1)$ otherwise, both within $g - 1 \leq g' \leq g + 1$.

        Now suppose $g = 1$, so that $p_i$ outputs $(x, 1)$ from $\ga2_1$.
        By graded delivery of $\ga2_1$, every \stable node $p_j$ outputs $(x, g_1)$ from $\ga2_1$ for $g_1 \leq 1$, so $p_j$ outputs $(x, 1)$ or $(x, 0)$ unless it obtains grade $1$ from $\ga2_0$.
        In the latter case, $p_j$ outputs $(y_0, 2)$ where $(y_0, 1)$ is its output from $\ga2_0$; by consistency of $\ga2_0$, every honest and awake node at round $6$ with an input $\neq \noinput$ inputs $y_0$ into $\ga2_1$, so by integrity of $\ga2_1$, no honest node outputs a value other than $y_0$ from $\ga2_1$, giving $y_0 = x$.
        Every \stable node $p_j$ thus outputs $(x, g')$ for $g' \leq 2$, again within $g - 1 \leq g' \leq g + 1$.
    \end{proof}

    \begin{lemma}[Consistency]
        \label{lem:immobile-ga3-consistency}
        If an honest node $p_i$ outputs $(x, g)$ for $g > 0$, then no honest node $p_j$ outputs $(x', \cdot)$ for $x' \neq x$.
    \end{lemma}

    \begin{proof}
        Nodes output the value of their $\ga2_0$ output at grade $2$ and the value of their $\ga2_1$ output at grades $1$ and $0$.
        First, suppose some honest node outputs $(y, 1)$ from $\ga2_0$.
        By consistency of $\ga2_0$, no honest node outputs a value other than $y$ from $\ga2_0$, so every honest and awake node at round $6$ with an input $\neq \noinput$ inputs $y$ into $\ga2_1$.
        By validity of $\ga2_1$, every \stable node outputs $(y, 1)$ from $\ga2_1$, so both branches only produce the value $y$, and $x = x' = y$.
        Otherwise, no honest node outputs grade $1$ from $\ga2_0$, so no honest node outputs grade $2$; since $p_i$ and $p_j$ only output because they have been awake since round $0$, both are \stable for both $\ga2_0$ and $\ga2_1$ and therefore obtain outputs from $\ga2_1$, with $p_i$ obtaining $x$ from grade $1$ of $\ga2_1$ and $p_j$ obtaining $x'$ from $\ga2_1$; by consistency of $\ga2_1$, we have $x = x'$.
    \end{proof}

    \begin{lemma}[Validity]
        \label{lem:immobile-ga3-validity}
        If every honest and awake node at round $t_s$ with an input $\neq \noinput$ inputs $x$, then every \stable node $p_i$ outputs $(x, 2)$.
    \end{lemma}

    \begin{proof}
        By validity of $\ga2_0$, every \stable node $p_i$ outputs $(x, 1)$ from $\ga2_0$.
        Therefore, node $p_i$ has $g_0 = 1$ and outputs $(x, 2)$.
    \end{proof}

    \begin{lemma}[Integrity]
        \label{lem:immobile-ga3-integrity}
        If no honest and awake node at round $t_s$ inputs $x$, then no honest node $p_i$ outputs $(x, \cdot)$.
    \end{lemma}

    \begin{proof}
        Node $p_i$ outputs $(x, 2)$ only with an output $(x, 1)$ from $\ga2_0$; by integrity of $\ga2_0$, no honest node outputs $x$ from $\ga2_0$, so no honest node outputs $(x, 2)$.
        Node $p_i$ outputs $(x, 1)$ or $(x, 0)$ only with an output $(x, \cdot)$ from $\ga2_1$.
        Recall that every honest node with an input $\neq \noinput$ into $\ga2_1$ inputs its output from $\ga2_0$; since no honest node outputs $x$ from $\ga2_0$, no honest and awake node at round $6$ inputs $x$ into $\ga2_1$.
        By integrity of $\ga2_1$, no honest node outputs $x$ from $\ga2_1$, so no honest node outputs $(x, 1)$ or $(x, 0)$.
    \end{proof}
\end{defer}

\begin{ifdeferred}{ga3proofs}
    Our agreement protocol in \cref{immobile:ba} requires a third grade.
    We obtain $3$-grade graded agreement, which we denote $\ga3$, by running two instances of $\ga2$ back to back, where each node inputs the value it outputs from the first instance into the second.
    This construction is standard, so we defer it and its analysis to \cref{sec:deferred-ga3proofs}, where \cref{thm:immobile-ga3} shows that $\ga3$ implements graded agreement (\cref{def:graded-agreement}) with $\gmax = 3$ for $\rec \geq 6$, $\sbl \geq 12$, and $\rho \geq 1$.
\end{ifdeferred}
\begin{defer}{ga3proofs}
    \begin{theorem}
        \label{thm:immobile-ga3}
        Let $t_s \geq 0$ be the start round.
        If every \vstable node has input $x \neq \noinput$ (very-stable input, \cref{def:immobile-ga-vstable-input}), then the protocol in \cref{fig:immobile-ga3} implements graded agreement (\cref{def:graded-agreement}) with start round $t_s$ and $\gmax = 3$ for schedules admissible in the fully fluctuating sleepy model with $\rec \geq 6$, $\sbl \geq 12$, and $\rho \geq 1$.
    \end{theorem}

    \begin{proof}[Proof (of \cref{thm:immobile-ga3})]
        Since $\rec \geq 6$ and $\sbl \geq 12$, both $\ga2_0$ and $\ga2_1$ have very-stable input and implement graded agreement by \cref{thm:immobile-ga2}.
        We then have \emph{graded delivery} by \cref{lem:immobile-ga3-gd}, \emph{consistency} by \cref{lem:immobile-ga3-consistency}, \emph{validity} by \cref{lem:immobile-ga3-validity}, and \emph{integrity} by \cref{lem:immobile-ga3-integrity}.
    \end{proof}
\end{defer}

\subsection{Byzantine Agreement}
\label{immobile:ba}

\begin{figure}[t]
    \centering
    \begin{tcolorbox}
        A node $p_i$ with input $x$ starts execution at round $0$ and executes the below steps at the start of the round $13v$ for every view $v \geq 0$.
        Beyond the outputs of $\gw$ and $\ga$, $p_i$ maintains the following state variables with the following initial values:
        \begin{itemize}
            \item \emph{Candidate: } $x_C \gets x$.
            \item \emph{Lock: } $x_L \gets \bot$.
        \end{itemize}

        \begin{protocolsteps}
            \item[$(t = 0 ... 12):$]Execute $\ga3_{v}$ with input $x_C$.
            \item[$(t = 12):$] If awake since the beginning of view $v$:\leavevmode
            \begin{itemize}
                \item Let $(y, g) \gets $ the output of $\ga3_{v}$; if $y \neq \bot$, set $x_C \gets y$.
                \item If $g \geq 1$: set $x_L \gets y$; otherwise, set $x_L \gets \bot$.
                \item If $g = 2$ and not yet output: output $y$.
                \item Multicast $\signedproposemsg[i]{(x_C, \vrfo{v+1}, \pi)}$.
            \end{itemize}
            \item[$(t = 13):$]\leavevmode
            \begin{itemize}
                \item If $x_L = \bot$: set $x_{C} \gets $ the proposed value in $\valspace$ with the highest value VRF output for view $v + 1$.
                \item If not awake since the beginning of the view: set $x_{C} \gets \noinput$.
            \end{itemize}
        \end{protocolsteps}
        Note that $t = 13$ of view $v$ occurs alongside $t = 0$ of view $v + 1$.
    \end{tcolorbox}
    \caption{Byzantine agreement in the fully fluctuating sleepy model.}
    \label{fig:immobile-ba}
\end{figure}

Our agreement protocol in \cref{fig:immobile-ba} follows the phase-king approach \cite{BGP89}.
The protocol proceeds in a series of views, each calling an instance of $\ga3$ from \cref{immobile:ga3}.
For a view $v \geq 0$, we denote $\ga3_{v}$ as the $3$-grade graded agreement instance of view $v$.

For the starting view $0$, awake nodes pass their initial input (called the \emph{candidate} input) to $\ga3_{0}$.
For a node that remains awake, it updates its candidate input to the output value of $\ga3_{0}$.
If the grade is at least $1$, it also locks on that value; otherwise, it releases its lock.
If the grade is $2$, it outputs the value in the agreement protocol (but continues to participate to help other nodes).
It then multicasts and signs its candidate input in $\valspace$, in a $\proposemsgtype$ message, along with its VRF output and proof over the next view number.
Such a proposal always exists, since at least one \stable node exists by admissibility with $\rho \geq 1$ and proposes a candidate in $\valspace$.

At the end of the view, nodes without a lock update their candidate input to the proposal with the highest VRF output for the view.
This input is then passed on to the next $\ga3$ instance.
The adversary may pre-compute VRF outputs for many future rounds and then cherry-pick the corrupt node with the highest VRF output to be awake for the corresponding view.
This hurts the latency of the agreement protocol but does not affect safety, liveness, or validity. 

We prove that the agreement protocol satisfies the properties of \cref{def:agreement} for schedules admissible in the $(\rec = 13, \sbl = 13, \rho = 1)$-fully fluctuating sleepy model.
First, we show in \cref{lem:immobile-ba-vstable-input} that each $\ga3_{v}$ satisfies the very-stable input condition (\cref{def:immobile-ga-vstable-input}) for \cref{thm:immobile-ga3}.

\begin{lemma}[Very-Stable Input]
    \label{lem:immobile-ba-vstable-input}
    For every view $v \geq 0$, each \vstable[$13v$] node has an input $x_{C} \in \valspace$.
\end{lemma}

\begin{defer}{baproofs}
    \begin{proof}[Proof (of \cref{lem:immobile-ba-vstable-input})]
        For the very first view $v = 0$, each \vstable[$0$] node has input $x_{C} \in \valspace$ for $\ga3_{0}$.
        For views $v > 0$, each \vstable[$13v$] node $p_i$ has been awake since round $13v - 13$, making it \stable[$13(v - 1)$] at the start of $\ga3_{v - 1}$.
        By definition, $\ga3_{v-1}$ outputs value $y \in \valspace \cup \{\bot\}$, and $p_i$ only sets its candidate $x_{C}$ to $y$ if $y \neq \bot$, so $x_C \in \valspace$.
        If $p_i$ holds no lock, it replaces its candidate with the proposal of the highest VRF output; at least one such proposal exists, since $p_i$ is itself \stable[$13(v - 1)$] and multicasts its own proposal at the end of view $v - 1$.
        Every proposal carries a value in $\valspace$, so each \vstable[$13v$] node, then, has an input $x_{C} \in \valspace$.
    \end{proof}
\end{defer}

\begin{lemma}[Totality]
    \label{lem:immobile-ba-totality}
    If an honest node $p_i$ outputs $y$ during view $v$, then for every view $v' > v$, every node in $\hons{13v'}{13v' + 12}$ outputs $y$ if it has not output already.
\end{lemma}

\begin{defer}{baproofs}
    \begin{proof}[Proof (of \cref{lem:immobile-ba-totality})]
        Since $p_i$ outputs $y$ during view $v$, it obtains $(y, 2)$ from $\ga3_{v}$.
        We first claim that, for any view $v' > v$, if some honest node obtains $(y, 2)$ from $\ga3_{v'-1}$, then every node in $\hons{13v'}{13v' + 12}$ obtains $(y, 2)$ from $\ga3_{v'}$.
        By graded delivery of $\ga3_{v'-1}$, every node in $\hons{13(v' - 1)}{13(v' - 1) + 12}$ obtains $(y, g)$ from $\ga3_{v'-1}$ for $g \geq 1$.
        Every such node therefore sets its candidate to $x_{C} = y$ and locks on $x_{L} = y$, so it keeps $x_{C} = y$ at round $13v'$, since a node holding a lock does not adopt a proposal.
        Any node honest and awake at round $13v'$ that was not awake throughout view $v' - 1$ resets its candidate to $\noinput$.
        Hence every honest and awake node at round $13v'$ with an input $\neq \noinput$ inputs $y$ into $\ga3_{v'}$, and by validity of $\ga3_{v'}$, whose very-stable input condition holds by \cref{lem:immobile-ba-vstable-input}, every $\hons{13v'}{13v' + 12}$ node obtains $(y, 2)$.
        The lemma then follows by induction on $v' > v$: for $v' = v + 1$, the claim applies with $p_i$ as the honest node, and for $v' > v + 1$, it applies with any node in $\hons{13v'}{13v' + 12}$, of which at least one exists since the schedule is admissible with $\rho \geq 1$, and which obtains $(y, 2)$ by the induction hypothesis.
    \end{proof}
\end{defer}

\begin{lemma}[Validity]
    \label{lem:immobile-ba-validity}
    If each honest node awake at round $0$ inputs $x$, then any node that outputs also outputs $x$.
\end{lemma}

\begin{defer}{baproofs}
\begin{proof}
    Since all honest nodes awake at round $0$ have input $x$, all \vstable[$0$] nodes input $x_{C} = x$ into $\ga3_0$.
    By validity of $\ga3$, every \stable[$0$] node obtains $(x, 2)$ from $\ga3_0$.
    Therefore, all \stable[$0$] nodes output $x$, and by \cref{lem:immobile-ba-agreement} (safety), any other honest node that outputs also outputs $x$.
\end{proof}
\end{defer}

\deferredproof{baproofs}{lem:immobile-ba-vstable-input,lem:immobile-ba-totality,lem:immobile-ba-validity}

\begin{lemma}[Safety]
    \label{lem:immobile-ba-agreement}
    For two honest nodes $p_i$ and $p_j$ that output $x$ and $x'$ at rounds $t$ and $t'$, respectively, $x = x'$.
\end{lemma}

\begin{proof}
    Without loss of generality, assume $t \leq t'$.
    Then we have two cases: they output in the same view ($t = t'$), or $p_j$ outputs in a later view ($t < t'$).
    For $t = t'$, we directly get agreement since no two honest nodes can output two different values with grade $2$ by consistency of $\ga3$.
    For $t < t'$, let views $v$ and $v'$ be the views in which $p_i$ and $p_j$ output, respectively.
    Since $p_j$ outputs in view $v'$, it is in $\hons{13v'}{13v' + 12}$, as only such nodes obtain a grade $2$ output.
    \Cref{lem:immobile-ba-totality} then gives that $p_j$ outputs $x$, so $x' = x$.
\end{proof}

\begin{lemma}[Liveness]
    \label{lem:immobile-ba-liveness}
    There exists a round $t_0$ such that for every $t \geq t_0$, every \vstable[$t$] node outputs.
\end{lemma}

\begin{proof}
    Any \vstable[$t$] node with $t$ in view $v$ satisfies $[t - \rec, t + \sbl] \supseteq [13v, 13v + 12]$, so it is \stable[$13v$] for view $v$.
    We call a view $v$ an \emph{honest-leader view} if the proposal with the highest VRF output for view $v + 1$ is from a node $p_\ell$ that is \stable[$13v$].
    We first show that if $v$ is an honest-leader view, then every \stable[$13(v+1)$] node outputs in view $v + 1$.

    Let $x$ be the candidate that $p_\ell$ proposes, which is its output value from $\ga3_{v}$ unless that value is $\bot$, in which case it is the candidate $p_\ell$ already held; either way $x \in \valspace$.
    A \stable[$13v$] node that holds no lock adopts $p_\ell$'s proposal, so its candidate is $x$.
    A \stable[$13v$] node holds a lock only if it obtains grade at least $1$ from $\ga3_{v}$, and then graded delivery of $\ga3_{v}$ gives every \stable[$13v$] node the same output value; since $p_\ell$ is one of them, that value is $x$, which the node keeps as its candidate.
    Either way every \stable[$13v$] node has candidate $x$ at round $13(v+1)$, while a node honest and awake at that round but not awake throughout view $v$ sets its candidate to $\noinput$.
    Every honest and awake node at round $13(v+1)$ with an input $\neq \noinput$ therefore inputs $x$ into $\ga3_{v+1}$, so by validity of $\ga3_{v+1}$, every \stable[$13(v+1)$] node obtains $(x, 2)$ and outputs $x$ in view $v + 1$.

    We now show that an honest-leader view $v_0$ occurs within $n\lambda$ views except with negligible probability.
    Fix a view $v$ with $t_s = 13v$.
    Since the schedule is admissible with $\rho \geq 1$, at least one \stable[$t_s$] node exists in view $v$.
    The adversary may cherry-pick any corrupt nodes since the beginning of the execution for the highest VRF output.
    Therefore, there are $|\hons{t_s}{t_s + 12} \cup \advs{0}{t_s + 13}|$ \emph{possible} proposals in the view, making the probability that an honest node's proposal has the highest VRF output be
    \[
        \frac{|\hons{t_s}{t_s + 12}|}{|\hons{t_s}{t_s + 12}| + |\advs{0}{t_s + 13}|}.
    \]

    In the worst case, only one node is honest and awake, and the remaining nodes are corrupt and asleep.
    Since VRF outputs are independent and uniform, the highest output is equally likely to belong to any node in $\nodes$.
    The probability that view $v$ is an honest-leader view is then at least $1/n$.
    Since the VRF outputs for each view are sampled anew, with the view number as the oracle input, this bound holds for a view regardless of the adversary's choices in the preceding views.
    The probability that none of $n\lambda$ consecutive views is an honest-leader view is therefore at most $(1 - 1/n)^{n\lambda}$, which is negligible in $\lambda$.

    Let $v_0$ be the first honest-leader view and set $t_0 = 13(v_0 + 1)$.
    Since the schedule is admissible with $\rho \geq 1$, at least one \stable[$13(v_0+1)$] node exists, so honest nodes indeed output in view $v_0 + 1$.
    Any \vstable[$t$] node with $t \geq t_0$ lies in some view $v \geq v_0 + 1$ and is \stable[$13v$].
    If $v = v_0 + 1$, it outputs by the implication above for the honest-leader view $v_0$.
    If $v > v_0 + 1$, it outputs by \cref{lem:immobile-ba-totality}, since honest nodes output in view $v_0 + 1$.
\end{proof}

\begin{theorem}
    \label{thm:immobile-ba}
    The protocol in \cref{fig:immobile-ba} implements Byzantine agreement (\cref{def:agreement}) for schedules admissible in the $(\rec = 13, \sbl = 13, \rho = 1)$-fully fluctuating sleepy model.
\end{theorem}

\begin{proof}
    Since each view $v \geq 0$ has very-stable input (\cref{def:immobile-ga-vstable-input}) by \cref{lem:immobile-ba-vstable-input}, $\rec \geq 6$, $\sbl \geq 12$, and $\rho \geq 1$, each $\ga3_{v}$ instance implements graded agreement by \cref{thm:immobile-ga3}.
    Thus, we have \emph{safety} by \cref{lem:immobile-ba-agreement}, \emph{validity} by \cref{lem:immobile-ba-validity}, and \emph{liveness} by \cref{lem:immobile-ba-liveness}.
\end{proof}

\section{Consensus in the Fully Fluctuating Sleepy Model with Uncorruption}
\label{mobile}

This section presents a discussion of our agreement protocol in \cref{immobile:ba} under an adversary that can uncorrupt corrupt nodes.
Uncorruption raises two challenges.
The first concerns both safety and liveness of the agreement protocol but is resolved with a minor change to $\gw3$ and $\ga2$.
The second addresses how the adversary may bias the VRF-based leader election, which requires a mild additional assumption to guarantee liveness.

\subsection{Challenge 1: Pre-Signed Equivocation}
\label{mobile:equivs-fix}

In the previous setting of \cref{immobile}, honest nodes are assumed to always have a ``clean slate,'' since a node made corrupt in a prior round remains corrupt.
With uncorruption, no node is guaranteed a clean slate.
In graded agreement, therefore, an honest node can appear as an equivocator in a view when its past corrupt-self pre-signed input messages for the view.
Honest and awake nodes are then unable to distinguish honest inputs from the pre-signed ones.
The graded agreement protocol loses validity since non-equivocating inputs may not include all honest inputs, so consequently, the agreement protocol loses safety and liveness.

The liveness violation follows since a set of non-equivocating inputters may never form a majority for a single input in $\ga2$ and subsequently in $\ga3$, so no node may ever output from $\ba$.
Safety is subtler: suppose an honest node decides $x$ in some view, causing the other honest and awake nodes to lock on $x$.
Then, in the next view, every honest inputter may be considered an equivocator, so $\ga3$ outputs $(\bot, 0)$ and honest nodes release their locks on $x$.
The adversary may then win the next leader election and propose $y \neq x$, which the honest nodes adopt and decide, violating agreement.

\begin{figure}[tb]
      \centering
      \begin{tcolorbox}
            A node $p_i$ \edit{with value $x$} executes the below steps starting from a round $t_s$.
            A node $p_j$ has \emph{acknowledged} a nonce $\nonce$ \edit{with value $x'$} once $p_i$ receives $\signedackmsg[j]{\nonce, \medit{x'}}$.
            \begin{protocolsteps}
                  \item[$(t = 0):$] %
                  Sample nonce $\nonces^{0} \overset{\$}{\gets} \{0, 1\}^{\lambda}$.
                  Multicast $\signedpingmsg[i]{\nonces^{0}}$.
                  \item[$(t = 3):$] %
                  Multicast $\signedackmsg[i]{\nonce, \medit{x}}$ for every nonce $\nonce \in \nonces^{3}$.
                  \item[$(t = 4):$] %
                  Output $\gawake[2]$ as the \edit{node-value pairs} that acknowledged every nonce in $\nonces^{2}$.
                  \item[$(t = 5):$] %
                  Output $\gawake[1]$ as the \edit{node-value pairs} that acknowledged every nonce in $\nonces^{1}$.
                  \item[$(t = 6):$] %
                  Output $\gawake[0]$ as the \edit{node-value pairs} that acknowledged $\nonces^{0}$.
            \end{protocolsteps}
      \end{tcolorbox}
      \caption{Graded wakeness with $3$ grades for a node $p_i$ with value $x$. Text in red indicates the changes to \cref{fig:gw3} that attach a node's value to each nonce acknowledgment.}
      \label{fig:mobile-gw3}
\end{figure}

Graded wakeness already solves the problem of distinguishing the latest messages from pre-signed ones.
We simply have $\gw3$ take in a value that nodes attach alongside their nonce acknowledgments.
\Cref{fig:mobile-gw3} shows the changes in red.
The graded agreement protocol then passes its input to $\gw3$ and is otherwise unchanged, except that a node counts as an inputter or an equivocator only for a value it acknowledged in $\gw3$, at the grade already used for that round in \cref{fig:ga2}.
A pre-signed input has no fresh nonce acknowledgment, so it counts neither as an input nor as equivocation evidence.
With these changes, the proofs of validity and integrity (\cref{lem:immobile-ga-validity,lem:immobile-ga-integrity}) remain unchanged since, except with negligible probability, no corrupt node can pre-sign nonce-value pairs.
Equivocation evidence remains transferable due to graded delivery of graded wakeness, so the proofs of graded delivery and consistency (\cref{lem:immobile-ga-gd,lem:immobile-ga-consistency}) also remain unchanged.

\subsection{Challenge 2: Biased Leader Election}
\label{mobile:leader-fix}

Uncorruption introduces a significant source of bias in the VRF-based leader election.
In the previous setting, the adversary could pull from its set of corrupt nodes since the beginning of execution as possible candidates for the leader election.
With uncorruption, the adversary could theoretically pre-compute VRF outputs for every node in $\nodes$ by corrupting and uncorrupting each node.
Therefore, the adversary could predict the outcome of every leader election.

In this extreme case, the protocol cannot guarantee liveness, so some additional assumption is needed.
In essence, any assumption that guarantees an eventual honest leader suffices.
We formalize this using the notion of an \emph{unpredictable} node in \cref{def:unpredictable}.
\begin{definition}[Unpredictable Nodes]
    \label{def:unpredictable}
    A node $p_i$ is said to be unpredictable for a view $v > 0$ if it is \stable[$13(v-1)$] and has never queried the VRF oracle $\ovrf$ for view $v$ before round $13(v - 1)$.
\end{definition}
The adversary does not know whether an unpredictable node in a view will win the leader election or not.
Unpredictability is a mild assumption in practice, since it is extremely unlikely that the adversary has obtained every node's VRF output before the election.
Some nodes might simply have never been corrupted, so they remain unpredictable.
Other nodes may be more reluctant to share VRF outputs with the adversary than signatures, since VRF outputs decide leader election, so the adversary obtains fewer of them.
Assuming such an unpredictable node occurs often enough, liveness is guaranteed.
\Cref{thm:mobile-ba} makes this precise.

\begin{theorem}
    \label{thm:mobile-ba}
    The protocol in \cref{fig:immobile-ba}, with $\gw3$ and $\ga2$ modified as in \cref{mobile:equivs-fix}, implements Byzantine agreement (\cref{def:agreement}) for schedules that:
    \begin{itemize}
        \item are admissible in the $(\rec = 13, \sbl = 13, \rho = 1)$-fully fluctuating sleepy model; and
        \item have at least $n\lambda$ views $v > 0$ for which some node is unpredictable (\cref{def:unpredictable}).
    \end{itemize}
\end{theorem}
The proof of liveness follows as outlined in \cref{lem:immobile-ba-liveness} because, except with negligible probability, eventually some unpredictable node will have the highest VRF output.
\section{Related Work}

\authcite{PS17} formalize the sleepy model.
Their original sleepy model assumes a fixed set of corrupt nodes that do not fluctuate.
The model has attracted significant attention, with many variants~\cite{DPS16,KRDO17,DGKR17,BGK+18} of the longest-chain approach of~\authcite{Nak09}, as well as works improving latency~\cite{MR22,GL23,ENRT26} and providing censorship resistance~\cite{ET25}.

The original sleepy model was extended in subsequent works~\cite{MMR23,DNTT23,DSTZ24} to allow for the corrupt set to grow, but not shrink, so these protocols are not fully fluctuating.
Recent work~\cite{TY26} gives tight upper and lower bounds on communication complexity in this model.
\authcite{ENP25} introduce the external adversary model, and, like PoSAT~\cite{DKT21}, utilize VDFs to achieve consensus under fully fluctuating participation.
\authcite{FLPA26} improve on the latency and security of fully permissionless PoW consensus.

\emph{Reconfigurable} consensus~\cite{LMZ09,JM14,LAB+06,SRMJ12,BSA14,DZ22,NNR26} also deals with participation changes over time but is not to be confused with the sleepy model.
Reconfiguration changes the set of nodes that \emph{can} participate in the protocol, which is orthogonal to a majority of the nodes being asleep at any given time.
The sleepy model allows for fluctuations in terms of which nodes \emph{do} participate.
Only a few works consider both reconfigurations and sleepy nodes~\cite{BGK+18,BHK+20,DPS16}.
\authcite{LR23b} show %
that any protocol in the reconfiguration setting must make use of non-time-malleable cryptography (e.g., VDFs).
Note that this result does not apply to our protocol since our model does not include reconfiguration.

Algorand~\cite{CM19} achieves consensus with \emph{player replaceability}, where each step of the protocol is executed by a different committee.
Algorand's committees are subsampled from the membership set via VRF-based sortition~\cite{MRV99}.
Algorand relies on a fixed threshold of honest and awake nodes, so it does not tolerate adversarially controlled sleepiness.
It requires a super-majority of honest and awake nodes, which more closely resembles traditional models.
Subsequent works~\cite{GHK+21,CGG+21,DDG+23} extend player replaceability to MPC and also require a super-majority of honest and awake nodes.

Works on mixed fault models~\cite{GP92,MNR19,ADKS22,GDCL22,AMN+20} consider a combination of crash and Byzantine faults simultaneously, but do not consider scenarios in which a majority of nodes may crash.
\authcite{AKSW22} implement storage primitives under \emph{churn}, where nodes continually enter and leave an asynchronous system.
They require a known threshold on the churn rate, so they also do not allow an arbitrary number of nodes to leave or crash at once. %

The mobile Byzantine fault model, formalized by~\authcite{G94}, captures an adversary that can reallocate its corruption budget across nodes, potentially \emph{curing} (i.e., uncorrupting) previously corrupt nodes.
The literature divides into constrained mobility, where the adversary moves with messages~\cite{BGH95}, and unconstrained mobility, where the adversary moves between rounds~\cite{G94}.
Another key axis is whether the cured node can detect that it was corrupt~\cite{G94,BSIW12} or may inadvertently continue sending corrupt messages~\cite{SYKM13,BDNP14}.
Our work more closely resembles the setting of unconstrained mobility with corruption detection.
Intuitively, the fully fluctuating sleepy model is unique in that nodes are already assumed to detect that they were asleep, so it is reasonable to assume they also detect that they were corrupt.

\section{Conclusion}

This work presents a consensus protocol in the fully fluctuating sleepy model without any hardware assumptions such as PoW or VDFs.
We introduce a new primitive, \emph{graded wakeness}, which we implement using several rounds of a simple interactive challenge-response protocol.
Graded wakeness allows us to upgrade the graded agreement of~\authcite{DSTZ24} to tolerate fully fluctuating participation.
We then give agreement protocols in the fully fluctuating sleepy model. 
Without uncorruption, no additional assumption is required.
With uncorruption, we make a minor change to the graded wakeness and graded agreement protocols and require an additional mild assumption.
We leave it for future work to optimize aspects of the protocols such as round and message complexity.

\subsection*{AI Disclosure}
We used Claude Code to improve the clarity, grammar, and readability of this work.

\bibliographystyle{plainurl}%
\bibliography{references}

\appendix
\deferredsection{ga3proofs}{Graded Agreement with 3 Grades}
\deferredsection{ga2proofs}{Deferred Proofs for Graded Agreement with 2 Grades}
\deferredsection{baproofs}{Deferred Proofs for Byzantine Agreement}

\end{document}